\documentclass[11pt]{amsart}
\usepackage{fullpage}
\usepackage[foot]{amsaddr}
\usepackage{microtype}
\usepackage[OT1]{fontenc}
\usepackage{eulervm}
\usepackage[tt=false]{libertine} 
\usepackage{bbold}
\usepackage{amsmath}
\usepackage{amssymb}
\usepackage{amsthm}
\usepackage{thmtools} 
\usepackage[linesnumbered,boxed,ruled,vlined]{algorithm2e}
\usepackage{algpseudocode}
\usepackage{enumitem}
\usepackage{multirow}
\usepackage{bm}
\usepackage{xifthen}
\usepackage{xspace}
\usepackage{tikz}
\usetikzlibrary{arrows.meta,positioning,calc,fit,backgrounds}

\usepackage[margin=1cm]{caption} 
\usepackage{subfig}

\usepackage[thinlines]{easytable}

\usepackage[bookmarks=true,hypertexnames=false,pagebackref]{hyperref}
\hypersetup{colorlinks=true, citecolor=blue, linkcolor=red, urlcolor=blue}

\usepackage{pgfplots}
\pgfplotsset{compat=1.16}
\usepackage{tikz}
\usetikzlibrary{arrows,arrows.meta,backgrounds,calc,fit,decorations.pathreplacing,decorations.markings,shapes.geometric}

\tikzstyle{internal} = [draw, fill, shape=circle]
\tikzstyle{external} = [shape=circle]
\tikzstyle{square}   = [draw, fill, rectangle]
\tikzstyle{triangle} = [draw, fill, regular polygon, regular polygon sides=3, inner sep=3pt]
\tikzstyle{pentagon} = [draw, fill, regular polygon, regular polygon sides=5, inner sep=2pt, minimum size=14pt]
\tikzset{every fit/.append style=text badly centered}

\usetikzlibrary{positioning,chains,fit,shapes,calc}
\usetikzlibrary{trees}
\usetikzlibrary{decorations.pathreplacing}
\usetikzlibrary{decorations.pathmorphing}
\usetikzlibrary{decorations.markings}
\tikzset{>=latex} 

\usepackage{ifthen}

\usepackage{cleveref}

\usepackage[textsize=tiny]{todonotes}

\usepackage[normalem]{ulem}

\usepackage{mleftright}

\usepackage{cool}
\Style{DSymb={\mathrm d},DShorten=true,IntegrateDifferentialDSymb=\mathrm{d}}

\renewcommand{\Pr}{\mathop{\mathrm{Pr}}\nolimits}

\def\*#1{\mathbf{#1}}
\def\+#1{\mathcal{#1}}
\def\-#1{\mathrm{#1}}
\def\=#1{\mathbb{#1}}
\def\^#1{\mathbb{#1}}

\newcommand{\norm}[1]{\ensuremath{\left\lVert #1\right\rVert}}
\newcommand{\abs}[1]{\ensuremath{\left\vert#1\right\vert}}

\newcommand{\eps}{\varepsilon}

\newtheorem{theorem}{Theorem}

\newtheorem{lemma}[theorem]{Lemma}

\newtheorem{proposition}[theorem]{Proposition}

\theoremstyle{definition}

\newtheorem{definition}[theorem]{Definition}

\theoremstyle{remark}
\newtheorem{remark}[theorem]{Remark}

\crefname{theorem}{Theorem}{Theorems}
\crefname{observation}{Observation}{Observations}
\crefname{claim}{Claim}{Claims}
\crefname{condition}{Condition}{Conditions}
\crefname{algorithm}{Algorithm}{Algorithms}
\crefname{property}{Property}{Properties}
\crefname{example}{Example}{Examples}
\crefname{fact}{Fact}{Facts}
\crefname{lemma}{Lemma}{Lemmas}
\crefname{corollary}{Corollary}{Corollaries}
\crefname{definition}{Definition}{Definitions}
\crefname{remark}{Remark}{Remarks}
\crefname{proposition}{Proposition}{Propositions}
\crefname{equation}{equation}{equations}
\crefname{enumi}{Case}{Case}
\creflabelformat{enumi}{(#2#1#3)}

\definecolor{HGcolor}{RGB}{255,50,50}

\makeatletter
\def\prob#1#2#3{\goodbreak\begin{list}{}{\labelwidth\z@ \itemindent-\leftmargin
      \itemsep\z@  \topsep6\p@\@plus6\p@
      \let\makelabel\descriptionlabel}
  \item[\textbf{Name}]#1
  \item[\textbf{Instance}]#2
  \item[\textbf{Output}]#3
  \end{list}}
\makeatother

\makeatletter
\providecommand\@dotsep{5}
\def\listtodoname{Todo list}
\def\listoftodos{\@starttoc{tdo}\listtodoname}
\makeatother

\renewcommand{\Tr}{\operatorname{Tr}}
\newcommand{\Sc}{\operatorname{Sc}}

\newcommand{\one}{\mathbf{1}}

\newcommand{\supp}{\operatorname{supp}}

\newcommand{\Tree}{\mathcal{T}}

\usepackage{nicefrac,comment}

\usepackage{silence}
\newboolean{doubleblind}
\setboolean{doubleblind}{false}

\begin{document}

\title{An $m+n^{3/2}$ algorithm for counting spanning trees by $\ell_1$-regularized resistance}

\author{Rong-Hua Li}
\author{Yichun Yang}
\address[Rong-Hua Li]{School of Computer Science, Beijing Institute of Technology, Beijing, China}
\address[Yichun Yang]{School of Computer Science, Beijing Institute of Technology, Beijing, China}

\maketitle

\begin{abstract}
  We study the basic problem of approximating the number of spanning trees of a graph. For a graph with $n$ vertices, $m$ edges, We propose
  an algorithm that approximates the number of spanning trees in $\widetilde O(m+n^{3/2}\eps^{-1})$ time.  Our algorithm improves upon the previously best known $\widetilde O(m+n^{15/8}\eps^{-7/4})$ time algorithm by Chu, Gao, Peng, Sachdeva, Sawlani, and Wang [FOCS 2018] and the $\widetilde O(m^{3/2}\eps^{-1})$ time algorithm by Liu, Peng, and Yang [FOCS 2026]. Notably, our algorithm is based on the novel concept of $\ell_1$-regularized resistance. We propose simple and efficient algorithm for computing $\ell_1$-regularized resistance and we show that they can be used to approximate the number of spanning trees by combining with the determinant sparsifier framework of Durfee, Peebles, Peng, and Rao [FOCS 2017]. Our algorithm matches the best known size of the determinant sparsifiers.
\end{abstract}

\section{Introduction}

Counting the number of spanning trees is a fundamental problem in theoretical computer science and graph theory. For a weighted undirected graph $G=(V,E,w)$ with $n$ vertices and $m$ edges, let
\[
  \Tree(G)=\sum_{T\text{ spanning tree of }G}\prod_{e\in T}w_e
\]
be its total spanning-tree weight. The classical Kirchhoff's matrix-tree theorem \cite{Kirchhoff1847} identifies
$\Tree(G)$ with any cofactor of the graph Laplacian.  Exact evaluation of $\Tree(G)$ therefore
reduces to computing the Laplacian determinant, which takes
$n^{\omega}$ arithmetic operations, where $\omega\approx 2.371$ is the
matrix-multiplication exponent~\cite{Almanetal2025}. Since the exact computation of the Laplacian determinant is expensive, the approximate counting problem is studied in previous researches.
Durfee, Peebles, Peng, and Rao firstly introduced determinant-preserving
sparsification and a recursive Schur-complement algorithm for this problem
\cite{DPPR}. They proposed a $\widetilde O(n^2\eps^{-2})$\footnote{The notation $\widetilde O$ hides polylogarithmic factors in $n,m,1/\eps$.} time algorithm for $\eps$-approximating the number of spanning trees. The same complexity can also be derived by combining the optimal sampling algorithm of spanning trees by Anari, Liu and Vuong \cite{ALV2022} with the classical simulated annealing \cite{SVV09}.

However, the question is whether we can approximate the number of spanning trees in subquadratic time. Following the determinant sparsifier framework \cite{DPPR}, the most expensive step is to estimate the effective resistance for all edges with high accuracy to construct the determinant sparsifier. Chu, Gao, Peng, Sachdeva, Sawlani, and Wang used short-cycle
decompositions to accelerate the resistance computation and obtained
$m^{1+o(1)}+n^{15/8+o(1)}\eps^{-7/4}$ runtime for approximate counting spanning trees
\cite{CGPSSW}. Recently, Liu, Peng, and Yang proposed a different
$\widetilde O(m^{3/2}\eps^{-1})$ algorithm based on identifying uncorrelated edge sets
\cite{LPY}. On expander graphs, Li and Sachdeva achieved
$\widetilde O(m+n^{3/2}\eps^{-1})$ runtime via random walk--based
local resistance representations~\cite{LiSachdeva}.
Their approach relies on constant spectral expansion
so that random walks mix in $\widetilde O(1)$ steps, which cannot be extended to general graphs. Our result reaches the same asymptotic bound on general graphs. Notably, the $n^{3/2}$ term is already optimal in the determinant sparsifier framework \cite{DPPR}. In addition, our algorithm goes beyond the $\widetilde O(m+n^{7/4})$ barrier predicted by \cite{LPY} under the possibly optimal $\widetilde O(m/\eps)$ time effective resistance estimation.

\subsection{Main result}

The following theorem states our main result.

\begin{theorem}\label{thm:main}
Let $G$ be a connected undirected graph with $n$ vertices, $m$ edges, and positive
edge weights satisfying $w_{\max}/w_{\min}\le n^{O(1)}$.  For
$n^{-O(1)}\le\eps\le1/2$, there is a randomized algorithm that returns
$\widehat{\Tree}(G)$ such that $(1-\eps)\Tree(G)\le\widehat{\Tree}(G)
\le(1+\eps)\Tree(G)$ with probability $2/3$. The expected running time is
\[
 \widetilde O\!\left(
   m+n^{3/2}\eps^{-1}\right).
\]
\end{theorem}

Next we explain the proof sketch of \Cref{thm:main} and several key ingredients of our proof.

\subsection{Estimating the resistance.}

We note that the main bottleneck of approximating the number of spanning trees in previous works is the effective resistance computation \cite{DPPR,CGPSSW,LiSachdeva}. The speedup of our algorithm comes from an improved approach to resistance estimation. Unlike previous algorithms, which estimate resistances directly, we introduce the novel notion of $\ell_1$-regularized resistance and obtain an efficient algorithm by estimating these $\ell_1$-regularized resistances. This idea is inspired by recent progress on $\ell_1$-regularized PageRank
\cite{WeiYang,martinez2023accelerated}. We believe that the study of $\ell_1$-regularized resistances will also be of independent interest.


\begin{definition}
    Given graph $G$ with Laplacian $L_G$ and its pseudo-inverse $L_G^\dagger$, let $b_{s,t}=e_s-e_t$, the effective resistance between two vertices $s$ and $t$ is defined as
    \[
        R_{s,t} = b_{s,t}^T L_G^\dagger b_{s,t}
    \]
\end{definition}


\begin{definition}
For a source $s$ and a $\ell_1$-regularized parameter $\lambda\in(1/n,1/2)$, consider the following optimization problem:
\[
 z_s^*=\arg\min_{z\ge0}
 \left\{\frac12z^TL_Gz-e_s^Tz+\lambda\one^Tz\right\}
\]
with the residual $q_s^*=e_s-L_Gz_s^*$. We define $z_s^*$ as the $\ell_1$-regularized potential, and 
\[
R_{s,t}^\lambda = b_{s,t}^T (z_s^*-z_t^*)
\]
 as the $\ell_1$-regularized resistance between $s$ and $t$.
\end{definition}

The Karush--Kuhn--Tucker (KKT) conditions imply
\[
 0\le q_s^*\le\lambda\one,\qquad \one^Tq_s^*=1,
 \qquad \abs{\supp z_s^*}\le\lambda^{-1}.
\]
The support bound depends only on the parameter $\lambda$. We exploit this locality to
compute the $\ell_1$-regularized resistance to high accuracy.



\begin{theorem}\label{thm:l1-resistance-error-bound}
    Given a graph $G$, accuracy parameter $\xi$, there is a randomized algorithm that takes $\widetilde O(m+n\lambda^{-2})$ preprocessing time to create a data structure of space $\widetilde O(n\lambda^{-1})$, from which one can query the $\ell_1$-regularized resistance $\hat{R}^\lambda_{u,v}$ for any vertex pair $(u,v)$ in $O(1)$ time such that $|\hat{R}^\lambda_{u,v}-R^\lambda_{u,v}|\le \xi$ with probability $1-n^{-O(1)}$.
\end{theorem}

We then estimate effective resistance using $\ell_1$-regularized
resistance. Although $R_{s,t}^\lambda$ need not approximate $R_{s,t}$
pointwise, we prove that the aggregate error guarantee holds. The following resistance sampling theorem is achieved by combining the estimation of $\ell_1$-regularized resistance and some refined estimation of the residuals.

\begin{theorem}\label{thm:l1-resistance-sampler}
Let $G=(V,E,w)$ satisfy the input conventions of \Cref{thm:main}, and
let $1/n<\lambda\le1/4$. In $\widetilde O(m+n\lambda^{-2})$ expected
preprocessing time, one can construct a randomized resistance-query
data structure. With probability $1-n^{-O(1)}$, for every pair $s\ne t$, a finite query distribution returning
$\widehat R_{s,t}$ such that
\[
 \frac{R_{s,t}}{16}\le\widehat R_{s,t}\le8R_{s,t},\qquad
 \sum_{e=(s,t)\in E}w_e\mathbb E_{\rm{query}}
 \frac{(\widehat R_{s,t}-R_{s,t})^2}{\widehat R_{s,t}}
 =O(\lambda^2n).
\]
The expectation is over query randomness. A query takes at most
$\widetilde O(\lambda^{-1})$ time, and its mean time is
$\widetilde O(1)$ when the edge is chosen with probability proportional
to (constant factor approximation of) $w_eR_e$. 
\end{theorem}


\subsection{Constructing the determinant sparsifiers}

We combine the resistance estimator of \Cref{thm:l1-resistance-sampler} with the determinant sparsidier framework by Durfee et. al \cite{DPPR}. The original result of \cite{DPPR} needs to compute the effective resistance between many pairs of vertices under high precision. We strenghen the variance analysis of \cite{DPPR} and show that the aggregate error resistance estimator by \Cref{thm:l1-resistance-sampler} is already enough to construct the determinant sparsifier and approximate the number of spanning trees.

\begin{definition}
    Given a graph $G=(V,E,w)$, a $\eps$-determinant sparsifier of $G$ is a (sparse) graph $H$ such that $(1-\eps)\Tree_G\le \Tree_H\le (1+\eps)\Tree_G$.
\end{definition}

\begin{theorem}[Determinant sparsification]\label{thm:determinant-sparsifier}
Let $G=(V,E,w)$ satisfy the input and accuracy conventions of
\Cref{thm:main}, and let $1/n<\lambda<1/2$. An $\eps$-determinant
sparsifier can be constructed with probability at least $2/3$ in
expected time
\[
 \widetilde O(m+n\lambda^{-2}+\lambda^2n^2\eps^{-2}
                                      +n^{3/2}\eps^{-1}).
\]
\end{theorem}

By extending the determinant sparsifier to the schur complement, we can approximate the number of spanning trees in the same asymptotic time as constructing the determinant sparsifier. Balancing the runtime by setting $\lambda=\eps^{1/2} n^{-1/4}$ we obtain the $\widetilde O(m+n^{3/2}\eps^{-1})$ time algorithm for approximate counting spanning trees, which proves \Cref{thm:main}.

\subsection{Some additional remarks.}
In an early version of this paper, the authors propose a high level idea: can we use $\ell_1$ regularization to localize the potential and therefore speed up the resistance computation and spanning tree counting? Working interactively with
GPT 5.6 Sol, this idea led to an $\widetilde O(m+n^{7/4})$ algorithm. The key part of the algorithm is to locally compute the regularized potential and guarantee the aggregate error of resistance estimation for spanning tree counting.

About one week later, the authors realized that the computation of
$\ell_1$-regularized resistances could be improved further for some steps. After interaction with GPT~6~Astra, this yielded an $\widetilde O(m+n^{3/2})$
algorithm.  The improved algorithm rests on two ingredients:
(i) using the optimal random spanning tree sampling to reduce the cost of computing
$\ell_1$-regularized resistances from $O(\lambda^{-3})$ dependence to $O(\lambda^{-2})$ dependence; and
(ii) using refined error estimates, which improve the aggregate Pearson error from
$O(\lambda n)$ to $O(\lambda^2 n)$.
Together, these improvements give the $\widetilde O(m+n^{3/2})$ time algorithm, which is
optimal within the determinant sparsification framework.  An open problem is
whether $n^{3/2}$ is also optimal for approximate spanning tree counting
in general.  The authors are responsible for independently checking and rewriting all the 
proofs.

\section{Preliminaries}\label{sec:prelim}

We work with a connected undirected weighted graph $G=(V,E,w)$, where
$n=|V|$, $m=|E|$, and all edge weights are positive.  Parallel edges may be
aggregated and self-loops are discarded.  For an edge
$e=(u,v)$, let $b_e=e_u-e_v$.  We write the Laplacian matrix $L=L_G=\sum_{e\in E}w_eb_eb_e^T$, the effective resistance $R_{u,v}=b_{u,v}^TL^\dagger b_{u,v}$.
Let $d_v=\sum_{(u,v)\in E}w_{uv}$, $d_{\max}=\max_vd_v$, and let
$w_{\min}$ be the minimum nonzero edge weight.  The leverage score of $e$ is
$\ell_e=w_eR_{u,v}$.  We record two classical identities.

\begin{lemma}[Foster's theorem \cite{Foster1949}]
\label{lem:foster}
For every edge $e\in E$,
$0<\ell_e\le1$, and
\[
 \sum_{e\in E}\ell_e=n-1.
\]
\end{lemma}

\begin{lemma}[Kirchhoff's matrix--tree theorem \cite{Kirchhoff1847}]
\label{lem:kirchhoff}
For every root $r\in V$,
\[
 \Tree(G)=\det L[V\setminus\{r\},V\setminus\{r\}].
\]
\end{lemma}

Let $V=F\sqcup T$ be a partition of the vertices such that the principal
submatrix $L_{FF}$ is invertible.  The \emph{Schur complement} of $L$ onto
$T$ is
\(
 \Sc(L,T):=L_{TT}-L_{TF}L_{FF}^{-1}L_{FT}.
\)
Block factorization of the Laplacian gives
\[
 \begin{pmatrix} L_{FF} & L_{FT} \\ L_{TF} & L_{TT} \end{pmatrix}
 =
 \begin{pmatrix} I & 0 \\ L_{TF}L_{FF}^{-1} & I \end{pmatrix}
 \begin{pmatrix} L_{FF} & 0 \\ 0 & \Sc(L,T) \end{pmatrix}
 \begin{pmatrix} I & L_{FF}^{-1}L_{FT} \\ 0 & I \end{pmatrix}.
\]
Taking any cofactor with root $r\in T$ and applying
\Cref{lem:kirchhoff} therefore yields the product identity
\[
 \Tree(G)=\det(L_{FF})\,\Tree(\Sc(L,T)).
\]

Moreover, the Schur complement is also a Laplacian matrix. It has the following property.
\begin{proposition}\label{prop:schur-original-resistance}
  For any $b\in \mathbb R^T$ with $b\perp \one$, we have $b^T\Sc(L,T)^\dagger b=b^T(L^\dagger)_{TT} b$. As a result, the effective resistance on Schur complement is same as original graph: $R^G_{u,v}=R^{\Sc(L,T)}_{u,v}$ for $u,v\in T$.
\end{proposition}

We also need the following two classical algorithmic results for our analysis.

\begin{lemma}[Nearly-linear SDDM solver \cite{SpielmanTeng,KMP}]\label{lem:sdd}
For an SDDM matrix $M$ with $q$ nonzeros, a vector $b$, and $0<\eta<1$, one can
compute $\widetilde x$ satisfying
\[
 \norm{\widetilde x-M^{-1}b}_M
 \le\eta\norm{M^{-1}b}_M
\]
in $\widetilde O(q\log(1/\eta))$ time, with probability $1-n^{-O(1)}$. 
\end{lemma}

\begin{lemma}[Resistance sketch \cite{SpielmanSrivastava}]
\label{lem:res-sketch}
After $\widetilde O(m)$ preprocessing, one can
answer every pair query $(u,v)$ in $\widetilde O(1)$ time with a value
$\widetilde R(u,v)$ satisfying
\[
 \tfrac12R_G(u,v)\le\widetilde R(u,v)\le2R_G(u,v)
\]
simultaneously for all $u,v$, with probability $1-n^{-O(1)}$.
\end{lemma}

\section{Estimating the $\ell_1$-regularized resistance}\label{sec:estimating-resistance}

In this section we prove \Cref{thm:l1-resistance-error-bound}.
For $s\in V$, write
\[
 \Phi_s(z)=\frac12z^TLz-e_s^Tz+\lambda\one^Tz,\qquad
 z_s^*=\operatorname*{argmin}_{z\ge0}\Phi_s(z),\qquad
 q_s^*=e_s-Lz_s^*.
\]

The KKT conditions give the following locality properties.

\begin{lemma}[Obstacle structure]\label{lem:obstacle-structure-main}
For $1/n<\lambda<1/2$, the optimizer $z_s^*$ is unique
and satisfies
\[
 0\le q_s^*\le\lambda\one,\qquad
 \one^Tq_s^*=1,\qquad
 (q_s^*)_v=\lambda\ \text{if }(z_s^*)_v>0,\qquad
 |\supp z_s^*|\le\lambda^{-1}.
\]
In particular, $s\in\supp z_s^*$ and $(q_s^*)_s=\lambda$.
\end{lemma}

\begin{proof}
The KKT conditions are
\[
 g=Lz-e_s+\lambda\one\ge0,\qquad z\ge0,\qquad z_vg_v=0.
\]
Since $q=e_s-Lz=\lambda\one-g$, they give $q\le\lambda\one$ and
$q_v=\lambda$ on $v\in \supp (z)$.  If $v\notin\supp (z)$ and $v\ne s$, then
$q_v=\sum_uw_{uv}z_u\ge0$.  The source $s\in \supp (z)$, since otherwise
$q_s=1+\sum_uw_{su}z_u\ge1>\lambda$.  In addition, we have $\one^Tq=\one^T(e_s-Lz)=\one^Te_s=1$.  Consequently,
$1=\sum_vq_v\ge\lambda|\supp (z)|$.
Finally, the existence and uniqueness of $z_s^*$ are standard for this convex quadratic program \cite[Sec. 4.2.5, 10.1]{BoydVandenberghe2004}.
\end{proof}

The following theorem gives high-accuracy $\ell_1$-regularized potentials. \Cref{thm:l1-resistance-error-bound} is a direct consequence of \Cref{thm:single-source-potential}.

\begin{theorem}\label{thm:single-source-potential}
After $\widetilde O(m)$ preprocessing of $G$, for any source $s$, accuracy
$\eta>0$, one can compute the $n$ vectors
$\{\widetilde z_s\}_{s\in V}$ in $\widetilde O(n\lambda^{-2})$ time such that, with probability $1-n^{-O(1)}$ for all $s\in V$,
\[
 \widetilde z_s\ge0,\qquad
 \supp\widetilde z_s\subseteq\supp z_s^*,\qquad
 |\supp\widetilde z_s|\le\lambda^{-1},\qquad
 \norm{\widetilde z_s-z_s^*}_\infty\le\eta.
\]
Moreover, $\widetilde q_s=e_s-L\widetilde z_s$ satisfies
$0\le\widetilde q_s\le2\lambda\one$ and $\one^T\widetilde q_s=1$.
\end{theorem}

We first use \Cref{thm:single-source-potential} to prove \Cref{thm:l1-resistance-error-bound}, then we prove \Cref{thm:single-source-potential}.

\begin{proof}[Proof of \Cref{thm:l1-resistance-error-bound}]
Run \Cref{thm:single-source-potential} for every $s\in V$ with
$\eta=\xi/4$ and failure probability $n^{-O(1)}$, and store every sparse column
$\{\widetilde z_s\}_{s\in V}$ in a static hash table.  The total preprocessing time is $\widetilde O(m+n\lambda^{-2})$ and
output space is $O(n\lambda^{-1})$. For a query $(u,v)$ return
\[
 \widehat R_{u,v}^\lambda
 =b_{u,v}^T(\widetilde z_u-\widetilde z_v).
\]
This only takes $O(1)$ time. The error satisfies
\[
 |\widehat R_{u,v}^\lambda-R_{u,v}^\lambda|
 \le |(\widetilde z_u-z_u^*)_u|+|(\widetilde z_u-z_u^*)_v|
     +|(\widetilde z_v-z_v^*)_u|+|(\widetilde z_v-z_v^*)_v|
 \le\xi.
\]
So we prove \Cref{thm:l1-resistance-error-bound}.
\end{proof}

We now prove \Cref{thm:single-source-potential}.  Our algorithm is based on finding a candidate set $F$ such that $ \supp z_s^*\subseteq F$, and then solves the original obstacle objective on
that candidate set $F$ using the flow-diffusion solver of \cite{ChenPengWang2021}.
Write $k=\lambda^{-1}$. After $\widetilde{O}(m)$ preprocessing, we prove that the process of finding $F$ and solving the obstacle objective on $F$ can be implemented in $\widetilde{O}(k^2)$ time for every $s\in V$. 

We begin with a lemma that shows how to find a candidate set $F$ such that $ \supp z_s^*\subseteq F$.

\begin{lemma}
\label{lem:idla-support-superset}
Let $A_s^*=\supp z_s^*$, and fix $0<\delta<1/2$.  Set
$r=\lceil32\log(2n/\delta)\rceil$ and $J=\lceil2rk\rceil$.
If $J\ge rn$, return $V$. Otherwise release $J$ particles sequentially
at $s$, each following the weighted random walk until it reaches a vertex
holding fewer than $r$ particles and then settling there.
The set $F$ of fully occupied vertices satisfies
\[
 \Pr[A_s^*\subseteq F]\ge1-\delta,\qquad |F|\le2k+1.
\]
The process needs only $J=\widetilde O(k)$ first-exit samples from a
connected growing set of at most $2k+1$ fully occupied vertices.
\end{lemma}

\begin{proof}
If $J\ge rn$, then $n\le J/r\le2k+1/r$, so the stated size bound
holds.  Otherwise every particle settles almost surely, since the
connected finite graph still has unused capacity.  The size bound follows
by counting particles.

Fix $v\in A_s^*$ and for
$u\in A_s^*$, put the hitting probability
\[
h_u=\Pr_u[\text{Random walk starting from } u \text{ hits }v\text{ before leaving }A_s^*].
\]  
Then $h_v=1$ and $h_u=\sum_{x\in A_s^*}\frac{w_{ux}}{d_u}h_x$ for $u\in A_s^*\setminus\{v\}$. Therefore $L_{A_s^*A_s^*}h=c_ve_v$ for some $c_v>0$. By KKT conditions we have $L_{A_s^*A_s^*}z_{s,A_s^*}^*=e_s-\lambda\one_{A_s^*}$. By multiply $h^T$ on both sides, we have
\[
h_s-\lambda\sum_{u\in A_s^*}h_u=h^T L_{A_s^*A_s^*}z_{s,A_s^*}^*=c_v(z_{s,A_s^*})_v=c_v(z_s^*)_v>0.
\]

We then consider two types of random walks. The type-$1$ is $J$ random walks starting from $s$. The type-$2$ is $r$ random walks starting from each $u\in A_s^*$.
Let $M_v$ count random walks starting from $s$ hitting $v$ before leaving
$A_s^*$, and let $W_v$ count all $r$ random walks starting from every $u\in A_s^*$ hitting $v$ before leaving $A_s^*$. Then we have
\begin{align*}
\mathbb EW_v&=r\sum_{u\in A_s^*}h_u\ge rh_v=r,\\
 \mathbb EM_v&=Jh_s\ge\frac{2r}{\lambda}h_s>2r\sum_{u\in A_s^*}h_u= 2\mathbb EW_v.
\end{align*}

Next, we exhibit a coupling of the two types of walks under which
$v\notin F$ implies $M_v\le W_v$.  Suppose a type-$1$ walk starting at $s$ reaches a vertex $u$ and occupies the $j$-th empty slot of $u$ for some $j\le r$.  Call the trajectory from $s$ to $u$ the \emph{real path}.  From that moment on, continue the walk
along the same trajectory as the $j$-th type-$2$ walk started at $u$; call
this continuation the \emph{ghost path}.  Under this coupling, if
$v\notin F$, every walk counted by $M_v$ is also counted by $W_v$:
\begin{enumerate}[label=(\roman*)]
\item If the real path of a type-$1$ walk reaches $v$, then, since
$v$ is not full, the walk stops at $v$ immediately.  The corresponding
ghost path therefore begins at $v$, and hence is counted in $W_v$.
\item If the real path stops at some $u$ without leaving $A_s^*$, and
the subsequent ghost path hits $v$ before leaving $A_s^*$, then this
ghost walk is again counted in $W_v$.
\end{enumerate}
In either case a walk contributing to $M_v$ injects a distinct ghost
slot counted by $W_v$, so $v\notin F$ implies $M_v\le W_v$. Since $M_v$ and $W_v$ are sums of independent Bernoulli variables, Chernoff bounds applied separately
to $M_v$ and $W_v$ give
\[
 \Pr[v\notin F]\le
 \Pr[M_v\le\tfrac34\mathbb EM_v]
 +\Pr[W_v\ge\tfrac32\mathbb EW_v]
 \le2e^{-r/32}.
\]
The independence between $M_v$ and $W_v$ is not required. A union bound over $A_s^*$ proves the result.
\end{proof}

We next discuss how to implement \Cref{lem:idla-support-superset} so as to
produce the candidate set $F$.  A direct simulation of the underlying random
walks would cost $\widetilde O(k^3)$ time: consider the case on an unweighted path,
once $s$ vertices have been filled, the expected time to reach a new vertex
is $\Theta(s^2)$.  Summing over the growth of $F$ therefore yields a $\Theta (k^3)$
bound. The $\Theta (k^3)$ runtime is not enough for the proof of \Cref{thm:single-source-potential}.

Therefore, we give an improved algorithm (\Cref{alg:wired-tree-full-set}) that
constructs $F$ in $\widetilde O(k^2)$ time. The key step of our algorithm is the optimal random spanning tree sampling by \cite{ALOGVV2021,ALV2022}. We notice that the random spanning tree satisfies the strong Rayleigh distribution. We define $T_\mu(t,r)$ for a distribution $\mu\in \mathbb R^{\binom{[m]}{r}}$ as the time it takes to produce a sample from $\mu$ conditional on all elements of the sample being a subset $T$ of size $|T|=t$. For strong Rayleigh distributions, \cite[Theorem 2]{ALV2022} reduce the sample from $\mu\in\mathbb R^{\binom{[m]}{r}}$ to the conditional sample of smaller size bounded by sums of its marginal probabilities.

\begin{algorithm}
\caption{generating the candidate set $F$}
\label{alg:wired-tree-full-set}
\KwIn{Weighted graph $G=(V,E,w)$, source $s\in V$, parameter $\lambda>0$, approximate leverage scores $\{u_e\}_{e\in E(G)}$}
\KwOut{Candidate set $F$}
$r\leftarrow\lceil32\log(2n/\delta)\rceil$,
$J\leftarrow\lceil2r/\lambda\rceil$\;
\If{$J\ge rn$}{\Return{$V$}\;}

Initialize $C\leftarrow\varnothing$ and a sparse occupancy map $a$ with default value $0$\;
\For{$j=1,\ldots,J$}{
  \eIf{$a(s)<r$}{
    $y\leftarrow s$\;
  }{
    Let $W_C$ be the implicit graph on $C\cup\{g\}$ with terminal node $g$. $W_C$ has internal edges $E(C)$ and external edges $(x,g)$ with weight $b_x=\sum_{v\notin C}w_{xv}$ for every $x\in C$\;
    Sample a weighted random spanning tree $T$ of $W_C$ by \Cref{lem:wired-tree-sampler}, to $TV$-distance $n^{-O(1)}$\;
    Let $(x,g)$ be the last edge of the unique $s$-to-$g$ path in $T$\;
    Sample $y\notin C$ with probability $w_{xy}/b_x$\;
  }
  $a(y)\leftarrow a(y)+1$\;
  \If{$a(y)=r$}{
    $C\leftarrow C\cup\{y\}$; update the implicit graph $W_C$\;
  }
}
\Return{$F\leftarrow C$}\;
\end{algorithm}

\begin{theorem}[{\cite[Theorem~2]{ALV2022}}]\label{thm:strong-rayleigh-sampling}
Let $\mu\in \mathbb R^{\binom{[m]}{r}}$ be a strongly Rayleigh distribution, and let $q_i\ge\Pr_{T\sim\mu}[i\in T]$ for $i\in [m]$ be the upper bound of the marginal probabilities which sums to $K=\sum_iq_i$. Sampling from $\mu$ with $TV$-distance $m^{-O(1)}$ reduces to $\widetilde O(1)$ calls to $T_\mu(O(K),r)$.
\end{theorem}

For weighted spanning trees, the nearly linear time sampler \cite[Theorem~2]{ALOGVV2021} gives $T_\mu(t,r)=\widetilde O(t)$, where $r=|V|-1$ is the edge count of a spanning tree. Therefore we have the following Lemma.

\begin{lemma}
\label{lem:wired-tree-sampler}
After $\widetilde O(m)$ shared preprocessing, the spanning-tree sampling
step of \Cref{alg:wired-tree-full-set} can be implemented in
$\widetilde O(|C|)$ expected time. 
\end{lemma}

\begin{proof}
Compute the resistance sketch $\ell_G(e)\le u_e\le c\ell_G(e)$, with $u_e\le1$, once
for all original edges using \Cref{lem:res-sketch}. This costs $\widetilde O(m)$ shared preprocessing time.
For every $e\in E(W_C)$, the property of leverage score gives
$\ell_{W_C}(e)\le\ell_G(e)\le u_e$ on internal edges. Assign upper
bound $1$ to each positive ground edge. Their total upper-bound mass is
\[
 \sum_{e\in E(C)}u_e+|\{x\in C:b_x>0\}|
 \le c\,\mathbb E_{T\sim\Tree_G}[|T\cap E(C)|]+|C|
 \le c(|C|-1)+|C|=O(|C|).
\]
Then \Cref{thm:strong-rayleigh-sampling} combined with the nearly linear time sampler \cite[Theorem~2]{ALOGVV2021} gives the $\widetilde O(|C|)$ time for sampling a weighted spanning tree.
\end{proof}

\begin{lemma}
\label{lem:wired-local-first-exit}
After $\widetilde O(m)$ shared preprocessing, the random set $F$ in
\Cref{lem:idla-support-superset} can be generated in
$\widetilde O(k^2)$ expected time, with total-variation error
at most $n^{-O(1)}$.
\end{lemma}

\begin{proof}
Run \Cref{alg:wired-tree-full-set}. The computation of resistance sketches costs $\widetilde O(m)$ by \Cref{lem:res-sketch}. If $J\ge rn$, return $V$. Otherwise we return a connected set $C$ with $|C|\le J/r\le2k+1$. 
For an exact weighted tree sample on $W_C$, the Wilson's algorithm \cite{Wilson1996} identifies
its $s$-to-$g$ path with a loop-erased random walk stopped at $g$.
Loop erasure preserves the last edge $(x,g)$, and sampling
$y\notin C$ with probability $w_{xy}/b_x$ recovers the original walk's
first exit. Thus the occupancy updates have exactly the distribution in
\Cref{lem:idla-support-superset}. By \Cref{lem:wired-tree-sampler},
sampling the tree and sampling $y$ take
$\widetilde O(|C|)$ time per query. The at most $J$ queries cost
$\widetilde O(Jk)$, and updating $W_C$ costs $\widetilde O(k^2)$. So the total expected time is $\widetilde O(k^2)$.
\end{proof}

\begin{theorem}[{\cite[Theorem~1.1]{ChenPengWang2021}}]
\label{thm:cpw-flow-diffusion}
Let $H$ be a connected weighted graph with $N$ vertices and $M_H$ edges,
and let $d\in\mathbb R^N$ satisfy $\one^Td\ge0$. Write
\[
 \Psi(x)=\tfrac12x^TL_Hx+d^Tx,\qquad
 \Psi^*=\min_{x\ge0}\Psi(x).
\]
For $0<\varepsilon<1$,
one can compute a feasible $x\ge0$ satisfying
$\Psi(x)\le\Psi^*/(1+\varepsilon)$ with high probability in
$\widetilde O(M_H)$ time.
\end{theorem}

\begin{lemma}[Rounding without losing the probability residual]
\label{lem:positive-residual-rounding}
Let $x\ge0$ satisfy $\|x-z_s^*\|_\infty\le\xi$, with
$\xi\le\lambda/(12d_{\max})$. Define
\[
 \widehat z_v=\begin{cases}x_v,&x_v>2\xi,\\0,&x_v\le2\xi.
 \end{cases}
\]
Then $\supp\widehat z\subseteq A_s^*$,
$\|\widehat z-z_s^*\|_\infty\le3\xi$, and
\[
 0\le e_s-L\widehat z\le\tfrac32\lambda\one,
 \qquad \one^T(e_s-L\widehat z)=1.
\]
\end{lemma}

\begin{proof}
A retained coordinate has $(z_s^*)_v>\xi>0$, so its exact residual
equals $\lambda$. A discarded coordinate has $(z_s^*)_v\le3\xi$.
Thus the infinity-norm error is at most $3\xi$, and
\[
 \|L(\widehat z-z_s^*)\|_\infty
 \le6d_{\max}\xi\le\lambda/2.
\]
Retained coordinates therefore have residual in
$[\lambda/2,3\lambda/2]$. At a discarded coordinate,
$\widehat z_v=0$ and the Laplacian sign pattern directly gives
$(e_s-L\widehat z)_v=\mathbf1_{v=s}+\sum_uw_{vu}\widehat z_u\ge0$;
the same error bound gives its upper bound. The sum follows from
$\one^TL=0$.
\end{proof}

\begin{algorithm}
\caption{Computing the $\ell_1$-regularized resistance sketch}
\label{alg:l1-resistance-sketch}
\KwIn{Connected weighted graph $G$, $1/n<\lambda<1/2$, potential accuracy $\eta$}
\KwOut{Sparse potentials $\{\widetilde z_s\}_{s\in V}$}
Precompute the leverage scores $\ell_G(e)\le u_e\le c\ell_G(e)$ for constant error $c$ by \Cref{lem:res-sketch}\;
\For{$s\in V$}{
  \tcp{Step 1: discover a candidate support.}
  $F\leftarrow$ \Cref{alg:wired-tree-full-set}$(G,s,\lambda, \{u_e\}_{e\in E(G)})$\;
  \tcp{Step 2: solve the original objective on $F$.}
  Construct a graph $H_F$ on $F\cup\{g\}$, with internal edges $E(F)$ and edge $(u,g)$ with weight $w_{ug}=\beta_u=(L_{FF}\one_F)_u$ \;
  Apply \Cref{thm:cpw-flow-diffusion} on $H_F$ and vector $d$ with $d_u=\lambda-\mathbf1_{u=s}$ for $u\in F$ and $d_g=1$. We solve a vector $x_s$ on $F\cup\{g\}$\;
  \tcp{Step 3: round while preserving a nonnegative residual.}
  Set $(\widetilde z_s)_v\leftarrow(x_s)_v$ if $(x_s)_v>2\xi$,
  and $(\widetilde z_s)_v\leftarrow 0$ otherwise\;
}
\Return{$\{\widetilde z_s\}_{s\in V}$ as the $\ell_1$-regularized resistance sketch}.
\end{algorithm}

Finally, we put things together and show that \Cref{alg:l1-resistance-sketch} is a $\widetilde{O}(m+n\lambda^{-2})$ time algorithm for computing the $\ell_1$-regularized resistance sketch.

\begin{proof}[Proof of \Cref{thm:single-source-potential}]
Consider the iteration for source $s$ in \Cref{alg:l1-resistance-sketch}.
By \Cref{lem:idla-support-superset,lem:wired-local-first-exit}, it finds
$F\supseteq A_s^*$ with $|F|=O(k)$ in $\widetilde O(k^2)$ time
after shared $\widetilde O(m)$ preprocessing. Next we apply \Cref{thm:cpw-flow-diffusion} to solve the obstacle objective on $H_F$ in $\widetilde{O}(nnz(L_{FF}))=\widetilde O(k^2)$ time. Then we round the solution  by \Cref{lem:positive-residual-rounding} to obtain the vector $\widetilde z_s$. As a result, computing all vectors $\{\widetilde z_s\}_{s\in V}$ takes $\widetilde{O}(m+n\lambda^{-2})$ time.

We next analyze the error. By \Cref{thm:cpw-flow-diffusion}, we have that the approximate solution $x_s$ satisfies
\[
 0\le\Phi_s(x_s)-\Phi_s(z_s^*)
 \le\frac{\eps}{1+\eps}|\Phi_s(z_s^*)|\le 2\eps\frac{n}{w_{\min}}
\]
since $\Phi_s(z_s^*)\le n/w_{\min}$. The variational inequality gives
\[
 \|x_s-z_s^*\|_L^2\le 4\eps\frac{n}{w_{\min}}.
\]

Next, by the norm inequality, we have that
\[
 \|x_s-z_s^*\|_\infty\le \sqrt{\frac{n}{w_{\min}}}\Vert x_s-z_s^*\|_L\le 2\sqrt{\eps} \left(\frac{n}{w_{\min}}\right)\le \xi
\]
by setting the accuracy $\eps=\frac{\xi^2}{4} \left(\frac{w_{\min}}{n}\right)^2$. Then we apply \Cref{lem:positive-residual-rounding} to obtain the vector $\widetilde z_s$. The resulting vector has
all guarantees of \Cref{thm:single-source-potential}, which finishes the proof.
\end{proof}

\section{From $\ell_1$-resistance to resistance sampler.}

In this section we prove \Cref{thm:l1-resistance-sampler}. We consider the original graph $G$, with Laplacian $L$. Fix the potentials
$\{\widetilde z_s\}_{s\in V}$ supplied by \Cref{thm:single-source-potential}
with accuracy $\eta$ sufficiently small, and write $\widetilde q_s=e_s-L\widetilde z_s$ and
$A_s=\supp\widetilde z_s$. Each $A_s$ is nonempty and has size at most
$\lambda^{-1}$. For $s\ne t$, put $b=b_{s,t}$ and $x=\widetilde z_s-\widetilde z_t$,
and define
\[
 E_{s,t}:=(\widetilde q_s-\widetilde q_t)^TL^\dagger(\widetilde q_s-\widetilde q_t),
 \qquad
 S_{s,t}:=x^T(\widetilde q_s-\widetilde q_t).
\]
Since $\widetilde q_s-\widetilde q_t=b-Lx$ and $b\perp\one$, expansion gives
\[
 E_{s,t}=R_{s,t}-2b^Tx+x^TLx,
 \qquad S_{s,t}=b^Tx-x^TLx.
\]
Consequently, we have that
\[
 R_{s,t}=b^Tx+E_{s,t}+S_{s,t}
        =R_{s,t}^\lambda+E_{s,t}+S_{s,t}+O(\eta).
\]
Here $|b^Tx-R_{s,t}^\lambda|\le4\eta$ by \Cref{thm:l1-resistance-error-bound}.
Our next goal is to estimate $E_{s,t}$ and $S_{s,t}$ separately and
add these estimates to estimate the resistance $R_{s,t}$.  Before we proceed, we first provide two useful lemmas.

\begin{lemma}\label{lem:l1-resistance-domination}
For $s\ne t$, $0\le R_{s,t}^\lambda\le R_{s,t}$.
\end{lemma}

\begin{proof}
Put $y=z_s^*-z_t^*$. By the optimal condition of $z_s^*$ and $z_t^*$, $(Lz_s^*-e_s+\lambda \one)^T(z_t^*-z_s^*)\ge 0$ and similarly $(Lz_t^*-e_t+\lambda \one)^T(z_s^*-z_t^*)\ge 0$. Adding these two inequalities gives $y^TLy\le b^Ty$. Therefore,
$0\le y^TLy\le b^Ty=R_{s,t}^\lambda$. By Cauchy--Schwarz,
\[
 (R_{s,t}^\lambda)^2=(b^Ty)^2
 \le (b^TL^\dagger b)y^TLy\le R_{s,t}R_{s,t}^\lambda,
\]
which proves the claim.
\end{proof}

\begin{lemma}\label{lem:l1-resistance-aggregate-error}
It holds that
\(
 \sum_{e\in E}w_e(R_e-R_e^\lambda)=\lambda n-1.
\)
\end{lemma}

\begin{proof}
Let $Z=[z_s^*]_{s\in V}$. Since $R_e^\lambda=b_e^TZb_e$,
$L=\sum_{e\in E}w_eb_eb_e^T$, and $q_s^*(s)=\lambda$ by the KKT conditions,
\[
 \sum_{e\in E}w_eR_e^\lambda
 =\Tr(LZ)=\sum_s(1-q_s^*(s))=n(1-\lambda).
\]
Subtracting this from Foster's identity $\sum_{e\in E}w_eR_e=n-1$ by \Cref{lem:foster} proves the claim.
\end{proof}

As a result of \Cref{lem:l1-resistance-aggregate-error}, the $\ell_1$-regularized resistance can estimate the resistance up to aggregate error $\lambda n$. However, the $O(\lambda n)$ error is not enough for the proof of \Cref{thm:l1-resistance-sampler}: our requirement is the $O(\lambda^2 n)$ aggregate Pearson error. In the following, we provide a more refined estimator to the resistance by estimating the error terms $E_{s,t}$ and $S_{s,t}$ separately. After that, the aggregate error will be bounded by $O(\lambda^2 n)$. The time complexity is asymptotically same as computing the $\ell_1$-regularized resistance.

\textbf{Estimating $E_{s,t}$.}
First we provide our algorithm for estimating $E_{s,t}$, see \Cref{alg:original-energy-sketch}. We analyze the accuracy and variance of the estimation.

\begin{algorithm}
\caption{Estimating $E_{s,t}$}
\label{alg:original-energy-sketch}
\KwIn{$G$, the potentials $\{\widetilde z_s\}_{s\in V}$, $0<\varepsilon_H<1/2$, and $d\ge1$}
\KwOut{the estimation $\widehat E_{s,t}$}
\tcp{Building the index $\mathcal T$}
Construct a spectral sparsifier $H$ with
$(1-\eps_H)L\preceq H\preceq(1+\eps_H)L$, such that $H=C^TWC$ with $|E(H)|=O(n\eps_H^{-2})$\;
\For{$j=1,\ldots,d$}{
  Draw $g_j\sim\mathcal N(0,I_{|E(H)|})$ independently\;
  Using \Cref{lem:sdd} to solve
  $H\phi_j=C^TW^{1/2}g_j$ with error $n^{-O(1)}$, shift $\phi_j$ so that $\phi_j\perp\one$\;
  Compute $a_j=H\phi_j$\;
  \For{$s\in V$}{
    Store $\mathcal T_{j,s}=\phi_j(s)-a_j^T\widetilde z_s$\;
  }
}
\Return{index $\mathcal T\in \mathbb R^{d\times n}$.}\\
\tcp{Given the index $\mathcal T$ and a query pair $s\ne t$, sample size $d_{s,t}\le d$}
\Return{$\widehat E_{s,t}=\frac1{d_{s,t}}\sum_{j=1}^{d_{s,t}}(\mathcal T_{j,s}-\mathcal T_{j,t})^2$.}
\end{algorithm}


\begin{lemma}
\label{lem:original-energy-unbiased}
For every fixed pair $s\ne t$ and sample count $1\le d_{s,t}\le d$,
\Cref{alg:original-energy-sketch} takes $\widetilde O(m+n\eps_H^{-2}d)$ time to build the index and takes $O(d_{s,t})$ query time to return $\widehat E_{s,t}$ that satisfies
\[
 \mathbb E_{\mathcal T}\widehat E_{s,t}
 =E_{s,t}+O\!\left(\varepsilon_HE_{s,t}
       +\varepsilon_H\sqrt{E_{s,t}\,x^TLx}
       +\varepsilon_H^2x^TLx\right),
\]
\[
 \operatorname{Var}_{\mathcal T}(\widehat E_{s,t})
 \le\frac{C}{d_{s,t}}
       \bigl(\|b-Hx\|_{H^\dagger}^2+\lambda R_{s,t}\bigr)^2.
\]
\end{lemma}

\begin{proof}
The runtime for building the spectral sparsifier is $\widetilde O(m+n\eps_H^{-2})$ by \cite{SpielmanSrivastava}. The runtime for building the index $\mathcal T$ is $\widetilde O(n\eps_H^{-2}d)$. Putting together, the total runtime is $\widetilde O(m+n\eps_H^{-2}d)$. We next prove the expectation. Without solver error, $\phi_j=H^\dagger C^TW^{1/2}g_j$ has mean zero
and covariance $H^\dagger HH^\dagger=H^\dagger$. Moreover,
\(
 \mathcal T_{j,s}-\mathcal T_{j,t}
 =\phi_j^T(b-Hx).
\)
Thus, without solver error,
\[
 \mathbb E_{\mathcal T}(\mathcal T_{j,s}-\mathcal T_{j,t})^2
 =(b-Hx)^TH^\dagger(b-Hx)
 =b^TH^\dagger b-2b^Tx+x^THx.
\]
For $r=b-Lx$ and $D=H-L$, we have $E_{s,t}=r^TL^\dagger r$ and
\[
 \mathbb E_{\mathcal T} \widehat E_{s,t}-E_{s,t}
 =r^T(H^\dagger-L^\dagger)r-2r^TH^\dagger Dx+x^TDH^\dagger Dx.
\]
The spectral bounds give
$|r^T(H^\dagger-L^\dagger)r|\le C\varepsilon_HE_{s,t}$,
$r^TH^\dagger r\le CE_{s,t}$, and
$x^TDH^\dagger Dx\le C\varepsilon_H^2x^TLx$.
By Cauchy--Schwarz for the cross term, 
\[|r^TH^\dagger Dx|=|(H^{\dagger/2}r)^T(H^{\dagger/2}Dx)|\le \|H^{\dagger/2}r\|_2\|H^{\dagger/2}Dx\|_2\le C\eps_H \sqrt{E_{s,t}x^TLx}.\]
 Therefore,
\[
 |\mathbb E_{\mathcal T} \widehat E_{s,t}-E_{s,t}|\le C\left(\varepsilon_HE_{s,t}
       +\varepsilon_H\sqrt{E_{s,t}\,x^TLx}
       +\varepsilon_H^2x^TLx\right).
\]

In the actual implenentation, we use \Cref{lem:sdd} to solve the linear system $H\phi_j=C^TW^{1/2}g_j$ with error $n^{-O(1)}$, we make the additional expectation error $O(\varepsilon_H^2R_{s,t})$ by setting the accuracy sufficiently small. This is absorbed by the preceding bound since $R_{s,t}=\|r+Lx\|_{L^\dagger}^2\le2E_{s,t}+2x^TLx$. Computing $\widehat E_{s,t}$ in the query process takes $O(d_{s,t})$ time.

Next we analyze the variance. With exact solves $\phi_j=H^\dagger C^TW^{1/2}g_j$ and $\mathcal T_{j,s}-\mathcal T_{j,t}=\phi_j^T(b-Hx)$, we have that $\mathcal T_{j,s}-\mathcal T_{j,t}\sim \mathcal N(0,\|b-Hx\|^2_{H^\dagger})$. So their squared values have variance
$2\|b-Hx\|^4_{H^\dagger}$. Hence
\[
 \operatorname{Var}_{\mathcal T}(\widehat E_{s,t})
 =\frac{2\|b-Hx\|^4_{H^\dagger}}{d_{s,t}}.
\]

The actual variance can be therefore bounded by $\frac{C}{d_{s,t}}\bigl(\|b-Hx\|_{H^\dagger}^2+\lambda R_{s,t}\bigr)^2$ when setting the accuracy of SDDM solver (\Cref{lem:sdd}) sufficiently small.
\end{proof}

\textbf{Estimating $S_{s,t}$.} Next we provide our algorithm for estimating $S_{s,t}$, see \Cref{alg:original-slack-query}. We analyze the accuracy and variance of the estimation.

\begin{algorithm}
\caption{Estimating $S_{s,t}$}
\label{alg:original-slack-query}
\KwIn{Stored potentials $\{\widetilde z_u\}$ with their supports $\{A_u\}$}
\KwOut{A sample $\widehat S_{s,t}$}
\tcp{Shared preprocessing part}
\For{$u\in V$}{
Set $\tau=\lambda/(64d_{\max})$. Consider the vector $\Psi_u$ with $\Psi_u(v)=\sqrt{\widetilde z_u(v)+\tau}-\sqrt\tau$. Round $\Psi_u(v)$ to a common dyadic grid of mesh at most $\sqrt\tau/8$\;
Store $\widetilde q_u(v)$ for $v\in A_u$ and compute the heavy sums $\sum_{a\in A_u:\,w_{va}\widetilde z_u(a)>\lambda^2}w_{va}\widetilde z_u(a)$\;
}
\tcp{Given a query pair $s\ne t$, sample size $d_{s,t}$}
\For{$j=1,\ldots,d_{s,t}$}{
Let $p_{s,t}\in \mathbb R^n$ such that $p_{s,t}(v)=\frac{(\Psi_s(v)-\Psi_t(v))^2}{\Vert \Psi_s-\Psi_t\Vert_2^2}$\;
Set $\pi_{s,t}(v)=\tfrac12p_{s,t}(v)
+\tfrac{\mathbf1_{\{v\in A_s\}}}{4|A_s|}
+\tfrac{\mathbf1_{\{v\in A_t\}}}{4|A_t|}$.
If $\Vert \Psi_s-\Psi_t\Vert_2=0$, instead use
$\pi_{s,t}(v)=\tfrac{\mathbf1_{\{v\in A_s\}}}{2|A_s|}
+\tfrac{\mathbf1_{\{v\in A_t\}}}{2|A_t|}$\;
Draw $v\sim\pi_{s,t}$ \;
\If{$v\in A_s\cap A_t$}{
  $X_j=(\widetilde z_s(v)-\widetilde z_t(v))
  (\widetilde q_s(v)-\widetilde q_t(v))/\pi_{s,t}(v)$\;
}
\Else{  
  Set $(i,u)=(s,t)$ if $v\in A_s\setminus A_t$, and $(i,u)=(t,s)$ otherwise\;
  Draw $a\sim A_u$ uniformly and set
  $Y_u(v)=\sum_{a\in A_u:\,w_{va}\widetilde z_u(a)>\lambda^2}
  w_{va}\widetilde z_u(a)+|A_u|w_{va}\widetilde z_u(a)
  \mathbf1_{\{w_{va}\widetilde z_u(a)\le\lambda^2\}}$\;
  $X_j=\widetilde z_i(v)
  (\widetilde q_i(v)-Y_u(v))/\pi_{s,t}(v)$
}
}
\Return{$\widehat S_{s,t}=\frac{1}{d_{s,t}}\sum_{j=1}^{d_{s,t}} X_j$}\;
\end{algorithm}

\begin{lemma}
\label{lem:original-slack-unbiased}
With unit-cost access to coordinate sampling to $p_{s,t}$,
after $\widetilde O(m+n\lambda^{-2})$ shared preprocessing,
\Cref{alg:original-slack-query} takes $O(d_{s,t})$ time to estimate $\widehat S_{s,t}$, satisfies
\[
 \mathbb E\widehat S_{s,t}=S_{s,t},\qquad
 \operatorname{Var}(\widehat S_{s,t})
 \le\frac{C\lambda^2}{d_{s,t}}
       \bigl(\|\Psi_s-\Psi_t\|_2^2+R_{s,t}\bigr)^2.
\]
Where $\Psi_u$ is defined in \Cref{alg:original-slack-query}, Line 2.
\end{lemma}

\begin{proof}
In the preprocessing process, we store $\widetilde q_u(v)$ for $v\in A_u$ and we compute all the heavy sums $\sum_{a\in A_u:\,w_{va}\widetilde z_u(a)>\lambda^2}w_{va}\widetilde z_u(a)$. We show that they can be stored in static dictionaries in $\widetilde O(n\lambda^{-2})$ time.
Indeed, for each fixed $u$, evaluating $\widetilde q_u$ on $A_u$ takes
$O(\lambda^{-2})$ time, so doing this for every source costs
$O(n\lambda^{-2})$.  For the heavy sums, fix $u$ and note that the pair $(v,a)$ with $v\notin A_u$, $a\in A_u$ satisfies
\[
 \sum_{a\in A_u}\sum_{v\notin A_u}w_{va}\widetilde z_u(a)
 =\sum_{v\notin A_u}\widetilde q_u(v)\le1,
\]
so at most $O(\lambda^{-2})$ such terms can exceed $\lambda^2$.  Computing
these heavy sums for every source $u\in V$ therefore also costs
$O(n\lambda^{-2})$.

Next we compute the expectation. For $v\notin A_u$, we have that
\[
 \mathbb E Y_u(v)
 =\sum_{b\in A_u:\,w_{vb}\widetilde z_u(b)>\lambda^2}
       w_{vb}\widetilde z_u(b)
 +\sum_{b\in A_u:\,w_{vb}\widetilde z_u(b)\le\lambda^2}
       w_{vb}\widetilde z_u(b)
 =\widetilde q_u(v).
\]
On an intersection coordinate $v\in A_s\cap A_t$ we compute $ \frac{(\widetilde z_s(v)-\widetilde z_t(v))((\widetilde q_s)_v-(\widetilde q_t)_v)}{\pi_{s,t}(v)}$ exactly. On either
boundary $v\in A_s\setminus A_t$ or $v\in A_t\setminus A_s$ its conditional expectation is
$(\widetilde z_s(v)-\widetilde z_t(v))
 (\widetilde q_s(v)-\widetilde q_t(v))$.
Consequently, the expectation satisfies
\[
 \mathbb E\widehat S_{s,t}
 =\sum_{v\in A_s\cup A_t}\pi_{s,t}(v)
   \frac{(\widetilde z_s(v)-\widetilde z_t(v))((\widetilde q_s)_v-(\widetilde q_t)_v)}{\pi_{s,t}(v)}
 =S_{s,t},
\]
by the definition of $S_{s,t}$. The query time is clearly $O(d_{s,t})$ given unit cost access to $p_{s,t}$.

For the variance, our goal is to prove the upper bound $|X_j|\le C\lambda\bigl(\|\Psi_s-\Psi_t\|_2^2+R_{s,t}\bigr)$.
To this end, we first lower bound $\pi_{s,t}(v)$ for $v\in A_s\cup A_t$. By the definition of $\pi_{s,t}(v)$, we have that
\[
\pi_{s,t}(v)\ge \max\left\{\frac{(\Psi_s(v)-\Psi_t(v))^2}{2\|\Psi_s-\Psi_t\|_2^2},\frac{1}{4(|A_s|+|A_t|)}\right\}.
\]
Thus for some small constant $c\le 1/4$ and $\tau\le \lambda/(64d_{\max})$,
\[
 \pi_{s,t}(v)\ge
 c\frac{(\Psi_s(v)-\Psi_t(v))^2+\tau\mathbf1_{\{v\in A_s\cup A_t\}}}
        {\|\Psi_s-\Psi_t\|_2^2+\tau(|A_s|+|A_t|)}.
\]
Next we upper bound $\widetilde z_i(v)$ and $Y_u(v)$. We show that for $v\in A_s\setminus A_t$, 
\[
\widetilde z_i(v)\le C((\Psi_s(v)-\Psi_t(v))^2+\tau)\qquad \text{and} \qquad Y_u(v)\le C\lambda \qquad \text{for constant } C.
\]
To see this, we assume $\Psi_i(v)=\sqrt{\widetilde z_i(v)+\tau}-\sqrt\tau$ without rounding. Solving $\widetilde z_i(v)$ gives $\widetilde z_i(v)\le C((\Psi_s(v)-\Psi_t(v))^2+\tau)$ since $\Psi_t(v)=0$. The same error order preserves after rounding. By the definition of $Y_u(v)$, we clearly have $Y_u(v)\le C\lambda$. Consequently,
\[
 |X_j|\le C\lambda\bigl(\|\Psi_s-\Psi_t\|_2^2
                              +\tau(|A_s|+|A_t|)\bigr)\le C\lambda(\|\Psi_s-\Psi_t\|_2^2+R_{s,t})
\]
for $v\in A_s\setminus A_t$. The case $v\in A_t\setminus A_s$ is symmetric. On $v\in A_s\cap A_t$, we have that $q_s^*(v)=q_t^*(v)=\lambda$ by the KKT conditions. For approximate solutions by \Cref{thm:l1-resistance-error-bound}, we have that for $u\in \{s,t\}$,
\[
|\widetilde q_u(v)-q_u^*(v)|= |L(\widetilde z_u-z_u^*)_v|\le 2d_{\max}\eta.
\]
As a result, $|\widetilde q_s(v)-\widetilde q_t(v)|\le 4d_{\max}\eta$. Since $\pi_{s,t}(v)\ge c\lambda/2$, and $|\widetilde z_s(v)-\widetilde z_t(v)|\le R_{s,t}+2\eta$ by \Cref{lem:l1-resistance-domination} and \Cref{thm:l1-resistance-error-bound}, we have that $|X_j|\le C\lambda R_{s,t}$, by setting the accuracy $\eta=O(\lambda /d_{\max})$. So we porve the variance part.
\end{proof}

\begin{remark}\label{remark:query-p}
In the actual implementation, explicitly constructing $p_{s,t}$ by scanning
$A_s\cup A_t$ takes $O(\lambda^{-1})$ time, so we do not explicitly construct it.
Instead, in preprocessing, we partition the coordinates of every $\Psi_u$
according to a balanced binary tree of depth
$\ell=\lceil\log n\rceil$, and store a JL sketch of every block with accuracy
$\delta=1/(100\ell)$. These sketches take
$\widetilde O(n\lambda^{-1})$ preprocessing time and space.
For a query $(s,t)$ at a block $I$, write
\(
 M(I)=\sum_{v\in I}(\Psi_s(v)-\Psi_t(v))^2.
\)
Subtracting the two block sketches and computing the squared norm gives
an estimate $\widehat M(I)$ satisfying
\(
 (1-\delta)M(I)\le\widehat M(I)\le(1+\delta)M(I)
\)
with high probability. If $I$ has children $I_1,I_2$, enter $I_1$ with probability
\(
 \frac{\widehat M(I_1)}{\widehat M(I_1)+\widehat M(I_2)},
\)
and otherwise enter $I_2$, recursively until reaching a vertex $v$. Then we have that for the actual sample probability $\widehat{p}_{s,t}(v)$,
\[
 \left(\frac{1-\delta}{1+\delta}\right)^\ell p_{s,t}(v)
 \le\widehat p_{s,t}(v)\le
 \left(\frac{1+\delta}{1-\delta}\right)^\ell p_{s,t}(v).
\]
Since $((1+\delta)/(1-\delta))^\ell\le e^{3\delta\ell}<3/2$, this gives
\(
 \tfrac12p_{s,t}(v)\le\widehat p_{s,t}(v)\le\tfrac32p_{s,t}(v).
\)
Each query to $\widehat p_{s,t}$ takes $\widetilde O(1)$ time. The constant-factor approximation of $p_{s,t}$ is enough for subsequent analysis.
\end{remark}

Using independent randomness for the two queries, define the estimation of resistance
\[
 \widehat R^{(0)}_{s,t}
 =b_{s,t}^T(\widetilde z_s-\widetilde z_t)
       +\widehat E_{s,t}+\widehat S_{s,t}.
\]
For our analysis, define the error upper bound $B_{s,t}$ as
\[
 B_{s,t}=E_{s,t}+\frac{\lambda}{\Lambda}\|\Psi_s-\Psi_t\|_2^2
                               +\lambda R_{s,t},\qquad \Lambda=1+\log\!\left(1+\frac{n d_{\max}}{\lambda w_{\min}}\right).
\]
We prove that the estimation variance can be bounded by a funtion of $B_{s,t}^2$.

\begin{lemma}
\label{lem:original-correction-variance}
Take $\varepsilon_H=c\sqrt\lambda$ and
$1\le d_{s,t}\le d$. For all $s\ne t$ with probability $1-n^{-O(1)}$,
\[
 \lambda R_{s,t}\le B_{s,t}\le C R_{s,t},\qquad E_{s,t}\le B_{s,t}.
\]
Moreover, the following equation holds for fixed index $\mathcal T$ with probability $1-n^{-O(1)}$:
\[
 \frac{\mathbb E_{\mathrm{query}}
       (\widehat R^{(0)}_{s,t}-R_{s,t})^2}{R_{s,t}}
 \le \frac{C(\Lambda+\log n)^2B_{s,t}^2}{d_{s,t}R_{s,t}}
       +C\bigl(\lambda E_{s,t}+\lambda^2R_{s,t}\bigr).
\]
\end{lemma}

\begin{proof}
The following property holds by our estimation:
\[
 b^Tx \le R_{s,t}+2\eta,\qquad
 x^TLx\le3R_{s,t},\qquad E_{s,t}\le C R_{s,t},\qquad
 \lambda\|\Psi_s-\Psi_t\|_2^2\le C R_{s,t}.
\]
The definition of $B_{s,t}$ gives $\lambda R_{s,t}\le B_{s,t}\le C R_{s,t}$ and $E_{s,t}\le B_{s,t}$. Moreover,
\[
 \lambda\|\Psi_s-\Psi_t\|_2^2+\lambda R_{s,t}\le\Lambda B_{s,t}.
\]
By \Cref{lem:original-energy-unbiased}, using
$\varepsilon_H=c\sqrt\lambda$ and the preceding bounds,
\[
 |\mathbb E_{\mathcal T}\widehat E_{s,t}-E_{s,t}|\le C\left(\sqrt{\lambda E_{s,t}R_{s,t}}
                              +\lambda R_{s,t}\right).
\]
In particular, $\mathbb E_{\mathcal T}\widehat E_{s,t}\le C(E_{s,t}+\lambda R_{s,t})\le CB_{s,t}$.

Use the Gaussian feature differences from the proof of
\Cref{lem:original-energy-unbiased}.
For every fixed $d_{s,t}\ge\log n$, by the variance bound of \Cref{lem:original-energy-unbiased}, we have the following with probability $1-n^{-O(1)}$:
\[
 |\widehat E_{s,t}-\mathbb E_{\mathcal T}\widehat E_{s,t}|
 \le C\mathbb E_{\mathcal T}\widehat E_{s,t}\left(\sqrt{\frac{\log n}{d_{s,t}}}
                           \right)+c\lambda R_{s,t}.
\]
So we can bound the error
\[
 \frac{(\widehat E_{s,t}-\mathbb E_{\mathcal T}\widehat E_{s,t})^2}{R_{s,t}}
 \le C \left( \frac{\log n\,B_{s,t}^2}{d_{s,t}R_{s,t}}
                                      +\lambda^2R_{s,t} \right).
\]
By \Cref{lem:original-slack-unbiased}, we have that
\[
 \operatorname{Var}(\widehat S_{s,t})
 \le\frac{C\lambda^2}{d_{s,t}}
            (\|\Psi_s-\Psi_t\|_2^2+R_{s,t})^2
 \le\frac{C\Lambda^2B_{s,t}^2}{d_{s,t}}.
\]
Therefore, we have that
\[
 \mathbb E_{\mathrm{query}}(\widehat R^{(0)}_{s,t}-R_{s,t})^2
 =\operatorname{Var}(\widehat S_{s,t})
                          +(\widehat E_{s,t}-E_{s,t})^2.
\]
Putting things together, we prove the lemma.
\end{proof}

Finally, we give a upper bound on sum of the error term $B_{s,t}$. Our proof will utilize the maximum value principle.

\begin{lemma}[Maximum value principle; cf.\ {\cite[Exercises~2.23--2.24]{LyonsPeres2016}}]
  \label{lem:maximum-value-principle}
  On a connected graph $G$, if a vector $p$ has at least one non-positive coordinate, and for all $a$ such that $p(a)>0$ satisfies $(Lp)(a)\le 0$, then $p\le 0$.
\end{lemma}

\begin{lemma}
\label{lem:original-budget-mass}
It holds that $\sum_{\{s,t\}\in E}w_{s,t}B_{s,t}=O(\lambda n)$.
\end{lemma}

\begin{proof}
To prove $\sum_{\{s,t\}\in E}w_{s,t}B_{s,t}=O(\lambda n)$, we bound the following two terms:
\[
\sum_{e}{w_eE_e}\le O(\lambda n)\qquad \text{ and }\qquad \sum_{e=(s,t)}{w_e\|\Psi_s-\Psi_t\|_2^2}\le O(n\Lambda).
\]

By definition, $E_{s,t}=\|b-Lx\|_{L^\dagger}^2$ with $x=\widetilde z_s-\widetilde z_t$. So $E_{s,t}= R_{s,t}-2R_{s,t}^\lambda+x^TLx$. When assuming $x=z_s^*-z_t^*$ is the exact solution, the optimal condition gives $x^TLx\le b^Tx=R_{s,t}^\lambda$. Therefore, $E_{s,t}\le R_{s,t}-R_{s,t}^\lambda$. By \Cref{lem:l1-resistance-aggregate-error}, we have that $\sum_{e}{w_eE_e}\le \lambda n$. For approxiamte solution $x=\widetilde z_s-\widetilde z_t$, we provided that $\sum_{e}{w_eE_e}\le O(\lambda n)$ by setting the error $\eta$ sufficiently small.

For the square-root term, we define row potentials $f_v$ such that $f_v(s)=z_s^*(v)$ for our analysis. They satisfy $Lf_v=e_v-\check q_v$,
where $\check q_v$ is a probability vector, and
$0\le f_v\le n/w_{\min}$. To see why $\check q_v$ is a probability vector, we prove two items: (i) $\check q_v(u)\ge 0$ for every $u\in V$; (ii) $\sum_{u\in V}\check q_v(u)=1$. The second item is obvious, since $\one^T \check q_v=\one^T (e_v-Lf_v)=1$. Our goal is to prove the first item. We prove $Lf_v-e_v\le 0$. We define a vector $u=\sum_{t}{w_{st}z_t^*}+e_s\ge 0$. By $Lz_t^*\ge e_t-\lambda \one$, we compute:
\begin{align*}
  Lu&\ge \sum_t w_{st}(e_t-\lambda\one)+Le_s\\
  &=\sum_t w_{st}e_t-d_s\lambda \one +d_s e_s-\sum_t w_{st}e_t=d_s(e_s-\lambda \one).
\end{align*}

We next compare $u$ with $d_sz_s^*$. We define $p=d_sz_s^*-u$. For all coordinate with $p(a)>0$, $z_s^*(a)>0$ since $z_s^*(a)\ge 0$ and $d_s>0$. Therefore, KKT condition gives $(Lz_s^*)(a)=e_s(a)-\lambda$, so 
\[(Lp)(a)=d_s(e_s(a)-\lambda)-(Lu)(a)\le 0.\] 

On the other hand, on the coordinate $z_s^*(a)=0$, the vector $p(a)\le 0$. By \Cref{lem:maximum-value-principle}, we have that $p\le 0$. Therefore, $d_sz_s^*\le u$. Especially, for coordinate $v$, 
\[
(Lf_v)(s)=d_sz_s^*(v)-\sum_t w_{st}z_t^*(v)\le \one_{\{s=v\}}.
\]

So we prove the residual  $\check q_v=e_v-Lf_v$ is a probability vector. Now we assume $\Psi_s(v)=\sqrt{f_v(s)+\tau}-\sqrt\tau$ without estimation error and rounding error. For positive $a,b$,
\[
 (\sqrt a-\sqrt b)^2
 \le\tfrac14(a-b)(\log a-\log b).
\]

Summation by parts therefore gives
\begin{align*}
 &\sum_{\{s,t\}\in E}w_{s,t}
    \bigl(\sqrt{f_v(s)+\tau}-\sqrt{f_v(t)+\tau}\bigr)^2\\
 &\qquad\le\tfrac14(Lf_v)^T\log(f_v+\tau\one)
 \le\tfrac14\log\!\left(1+\frac{n}{\tau w_{\min}}\right)
 =O(\Lambda).
\end{align*}

Sum over $v$ gives $\sum_{e=(s,t)}{w_e\|\Psi_s-\Psi_t\|_2^2}\le O(n\Lambda)$. Next we consider $\Psi_s$ with estimation error and rounding error. By computation, the estimation error $\eta$ and rounding error $\tau$ will affect the total sum at most:
\[
 O\!\left(\frac{n d_{\max}}{\lambda}
                 \left(\frac{\eta^2}{\tau}+\tau\right)\right)=O(n),
\]
which is dominated by $O(n\Lambda)$. Finally, Foster's identity gives
$\sum_{\{s,t\}\in E}w_{s,t}R_{s,t}=n-1$. Putting things together,
\[
 \sum_{\{s,t\}\in E}w_{s,t}B_{s,t}
 =O(\lambda n)+\frac{\lambda}{\Lambda}O(n\Lambda)+\lambda(n-1)
 =O(\lambda n),
\]
which proves the lemma.
\end{proof}

\begin{proof}[Proof of \Cref{thm:l1-resistance-sampler}]
Use $\varepsilon_H=O(\sqrt{\lambda})$ from \Cref{lem:original-correction-variance}.
We first construct a constant-factor estimation of the error upper bound
$B_{s,t}$, which will be used to determine the sample count $d_{s,t}$.
By \Cref{lem:res-sketch}, we compute a resistance sketch $R^G_{s,t}$ such that $R_{s,t}\le R^G_{s,t}\le4R_{s,t}$, each $R^G_{s,t}$ can be queried in $\widetilde O(1)$ time. By JL-sketch of $\{\Psi_s\}$ (see \Cref{remark:query-p}), we query a constant-factor estimate $Q_{s,t}$ of
$\|\Psi_s-\Psi_t\|_2^2$ in $\widetilde O(1)$ time.

For the estimation of $E_{s,t}$, we choose $d_{s,t}=O(\log n)$ and return $U_{s,t}=\widehat E_{s,t}$ in $\widetilde O(1)$ time. The expectation and variance analysis by \Cref{lem:original-energy-unbiased} guarantee that this is a constant factor estimation of $(E_{s,t}+\lambda R_{s,t})$. As a result, we can estimate a constant factor estimation of $B_{s,t}$ in $\widetilde O(1)$ time:
\[
 \widetilde B_{s,t}
 =U_{s,t}+\frac{\lambda}{\Lambda}Q_{s,t}
                              +\lambda R^G_{s,t}.
\]

Now we choose the sample count $d_{s,t}$ for \Cref{alg:original-slack-query,alg:original-energy-sketch} based on $\widetilde B_{s,t}$. We set
\(
 d=\left\lceil\frac{C(\Lambda+\log n)^2}{\lambda}\right\rceil
\)
and
\(
 d_{s,t}=\min\left\{d,\left\lceil
 \frac{C(\Lambda+\log n)^2\widetilde B_{s,t}}
      {\lambda R^G_{s,t}}\right\rceil\right\}.
\)
Here $C$ is a sufficiently large constant. We then compute the estimation:
\begin{align*}
 \widehat R^{(0)}_{s,t}
       &=b_{s,t}^T(\widetilde z_s-\widetilde z_t)
                         +\widehat E_{s,t}+\widehat S_{s,t},\\
 \widehat R_{s,t}&=\min\{2R^G_{s,t},\max\{R^G_{s,t}/16,\widehat R^{(0)}_{s,t}\}\}.
\end{align*}
By our setting of $d_{s,t}$, we have that
\(
 \frac{(\Lambda+\log n)^2B_{s,t}^2}{d_{s,t}R_{s,t}}
 \le C\lambda B_{s,t}.
\)
Therefore, \Cref{lem:original-correction-variance} and \Cref{lem:original-budget-mass} give
\begin{align*}
 \sum_{\{s,t\}\in E}w_{s,t}
       \frac{\mathbb E_{\mathrm{query}}(\widehat R^{(0)}_{s,t}-R_{s,t})^2}{R_{s,t}}
 &\le C\lambda\sum_{\{s,t\}\in E}w_{s,t}B_{s,t}
       +C\lambda\sum_{\{s,t\}\in E}w_{s,t}E_{s,t}\\
 &\quad{}+C\lambda^2\sum_{\{s,t\}\in E}w_{s,t}R_{s,t}
 =O(\lambda^2n).
\end{align*}
By our setting of $\widehat R_{s,t}$, we have that
$R_{s,t}/16\le\widehat R_{s,t}\le8R_{s,t}$.
Replacing the denominator by $\widehat R_{s,t}$ costs at most a factor
of $16$, proving the aggregate-error assertion.

A query takes $\widetilde O(d_{s,t})$ time. Its worst-case time is $\widetilde O(d)=\widetilde O(\lambda^{-1})$.
For an edge sampled with probability $w_eR_e/(n-1)$,
\begin{align*}
 \mathbb E_e d_e
 &=\frac1{n-1}\sum_{e\in E}w_eR_e d_e\\
 &\le 1+\frac{C(\Lambda+\log n)^2}{\lambda(n-1)}
                \sum_{e\in E}w_eB_e
 =O\bigl(1+(\Lambda+\log n)^2\bigr).
\end{align*}
Thus the weighted mean query time is $\widetilde O(1)$.
Constructing the index in \Cref{alg:original-energy-sketch} costs $\widetilde O(m+n\eps_H^{-2}d)=\widetilde O(m+n\lambda^{-2})$. Preprocessing process of \Cref{alg:original-slack-query} costs $\widetilde O(m+n\lambda^{-2})$. Computing the $\ell_1$-regularized resistance sketch takes $\widetilde O(m+n\lambda^{-2})$ time. So the total preprocessing time is $\widetilde O(m+n\lambda^{-2})$.
\end{proof}

\section{Constructing the determinant sparsifiers}\label{sec:det-sparsifier}

In this section, we use the approximate resistances of
\Cref{thm:l1-resistance-sampler} to construct determinant sparsifiers, and
thereby prove \Cref{thm:determinant-sparsifier}.  We then extend the same
construction to implicit Schur complements in the next section, and combine it with the
determinant-preserving recursion of \cite{DPPR} to obtain an algorithm for
approximate spanning-tree counting, in the same asymptotic time as building the
sparsifier.

\subsection{Determinant sparsification from resistance estimation}

In this subsection we prove \Cref{thm:determinant-sparsifier}.  Let $S=n-1$ be the number of edges in a
spanning tree; Foster's identity gives $\sum_{e\in E}\ell_e=S$ for the true
leverage scores $\ell_e=w_eR_e$.  For the determinant sparsifier, we would like
approximate leverage scores $a_e=w_e\widehat R_e$ satisfying the aggregate
Pearson bound
\[
E_0:=\sum_{e\in E}\frac{(\ell_e-a_e)^2}{a_e}=O(\lambda^2 n).
\]

However, the guarantee of \Cref{thm:l1-resistance-sampler} does not provide such aggregate error guarantee: the aggregate error bound holds only in expectation over the
fresh query randomness.  To resolve this, we treat the random query outputs as
virtual parallel edges.  Specifically, for each edge $e$, let the query output be
marked by $\omega$ and occur with probability $\nu_e(\omega)$.  Conceptually
split $e$ into parallel copies $(e,\omega)$ with weights
$w_{e,\omega}=w_e\nu_e(\omega)$, and set
\[
 \ell_{e,\omega}=w_e\nu_e(\omega)R_e,\qquad
 a_{e,\omega}=w_e\nu_e(\omega)\widehat R_e(\omega).
\]
The resulting virtual multigraph has the same Laplacian $L$ as $G$, and
\[
 \sum_{e,\omega}\frac{(a_{e,\omega}-\ell_{e,\omega})^2}{a_{e,\omega}}
 =\sum_{e\in E}w_e\sum_{\omega}\nu_e(\omega)
  \frac{(\widehat R_e(\omega)-R_e)^2}{\widehat R_e(\omega)}
 =O(\lambda^2 n).
\]
Thus we may informally first assume that the approximate resistances themselves satisfy
the aggregate Pearson error $E_0$. The rigorous argument will be provided later.

The determinant sampling framework of \cite{DPPR} samples edges from a distribution with total mass $S$, matching the normalization of the true leverage scores. In general $\sum_ea_e$ need not equal $S$, so we set $A=\sum_ea_e$ and renormalize to the sampling masses $r_e:=Sa_e/A$, for which $\sum_er_e=S$. The deviation of this distribution from the true leverage scores is measured by the Pearson error
\[
\Xi(r\Vert\ell):=\sum_{e\in E}\frac{(r_e-\ell_e)^2}{r_e},
\]
which controls the second-moment bound below. The next lemma bounds this Pearson error after renormalization.

\begin{lemma}
\label{lem:renormalization-pearson}
Let $A=\sum_ea_e, S=n-1$ and $r_e=Sa_e/A$.  If for every edge $e$,
\[
\sum_e(\ell_e-a_e)^2/a_e\le E_0\qquad \text{ and }\qquad
\ell_e/16\le a_e\le8\ell_e,
\]
then $A=\Theta(S)$, $r_e\ge\ell_e/C$ for an absolute constant $C$, and $\Xi(r\Vert\ell)=O(E_0)$.
\end{lemma}

\begin{proof}
Since $\ell_e/16\le a_e\le8\ell_e$ and $\sum_e\ell_e=S$, we have
$S/16\le A\le8S$.
Foster's identity and Cauchy--Schwarz give
\[
 |A-S|
 =\left|\sum_e(a_e-\ell_e)\right|
 \le\left(\sum_e\ell_e\right)^{1/2}
       \left(\sum_e\frac{(a_e-\ell_e)^2}{\ell_e}\right)^{1/2}.
\]
Using $a_e\le8\ell_e$ in the second factor yields
$|A-S|\le\sqrt{8SE_0}$.
The bounds on $A$ give $A=\Theta(S)$ and hence
$r_e=Sa_e/A\ge\ell_e/C$ for some absolute constant $C$.

Put $c=S/A$.  Since $r_e-\ell_e=c(a_e-\ell_e)+(c-1)\ell_e$,
\[
 \Xi(r\Vert\ell)
 \le 2c\sum_e\frac{(a_e-\ell_e)^2}{a_e}
     +2(c-1)^2c^{-1}\sum_e\frac{\ell_e^2}{a_e}.
\]
The first term is $O(E_0)$.  For the second, $\sum_e\ell_e^2/a_e\le16S$, $c^{-1}=\Theta(1)$ and
$(c-1)^2=(A-S)^2/A^2= O(E_0/S)$, so the second term is also $O(E_0)$.
\end{proof}

We next provide the variance analysis of determinant sparsification. Our proof strengthens the result in \cite{DPPR} from pointwise leverage-score approximation to an aggregate Pearson-error guarantee. Since the proof is complicated and will distract us from the main argument, we defer it to \Cref{sec:missing-proofs}.

\begin{lemma}[Aggregate-error determinant sampling]
\label{lem:aggregate-determinant-sampling}
Let $r_e>0$ satisfy $\sum_er_e=S$ and $r_e\ge\ell_e/C$.  Independently draw
$t\ge4S$ edges with probabilities $p_e=r_e/S$.  Give every sampled copy of
$e$ weight
$w'_e=(w_eS/r_e)(1/(t)_S)^{1/S}$, where
$(t)_S=t(t-1)\cdots(t-S+1)$.
If $H$ is the resulting weighted multigraph, then the expectation and second moment satisfy
\[
 \mathbb{E}[\Tree(H)]=\Tree(G),\qquad
 \frac{\mathbb{E}[\Tree(H)^2]}{\Tree(G)^2}
 \le\exp\!\left(
 O_C\!\left(\frac{S\Xi(r\Vert\ell)}t+
                   \frac{n^3}{t^2}\right)\right).
\]
\end{lemma}

By \Cref{lem:aggregate-determinant-sampling}, we can construct the determinant sparsifier by sampling from the approximate leverage score $r_e$. However, there is one additional wrinkle: computing $r_e$ requires computing $A$. Directly computing $A$ is not possible by our resistance sampler. So we provide the following lemma to compute $A$. We defer its proof to \Cref{sec:missing-proofs}.

\begin{lemma}
\label{lem:original-normalizer}
Put $S=n-1$ and $A=\sum_{e\in E}w_e\mathbb E_{\mathrm{query}}\widehat R_e$.
For every $0<\alpha<1/10$, in
$\widetilde O(\lambda^2n^2\alpha^{-1}+n\alpha^{-1/2})$ expected time,
one can compute $\overline A\ge S/32$, independently of the determinant
samples, such that
\[
 Z=(\overline A/A)^S,\qquad
 e^{-\alpha}\le\mathbb EZ\le e^\alpha,\qquad
 \frac{\mathbb EZ^2}{(\mathbb EZ)^2}\le e^\alpha.
\]
\end{lemma}

We note that in \Cref{lem:original-normalizer}, guarantee $Z=(\overline A/A)^S$ is enough for the determinant sparsification. This is because every spanning tree has $n-1$ edges, and the multiplicative error produced by weighting the edges is at most $(\overline A/A)^S$ for every spanning tree. 

\begin{proof}[Proof of \Cref{thm:determinant-sparsifier}]
Given $e\in E$, for a query mark $\omega$ of probability $\nu_e(\omega)>0$, split
$e$ into a virtual edge of weight $w_e\nu_e(\omega)$ and define
\[
 \ell_{e,\omega}=w_e\nu_e(\omega)R_e,\qquad
 a_{e,\omega}=w_e\nu_e(\omega)\widehat R_e(\omega),\qquad
 r_{e,\omega}=Sa_{e,\omega}/A.
\]
This leaves the Laplacian and tree count unchanged. Moreover,
\[
 \sum_{e,\omega}\frac{(a_{e,\omega}-\ell_{e,\omega})^2}{a_{e,\omega}}
 =\sum_e w_e\mathbb E
       \frac{(\widehat R_e-R_e)^2}{\widehat R_e}
 =O(\lambda^2n).
\]
This is a deterministic bound on the virtual graph, so
\Cref{lem:renormalization-pearson} gives
$A=\Theta(S)$, $r_{e,\omega}\ge\ell_{e,\omega}/C$, and
$\Xi(r\Vert\ell)=O(\lambda^2n)$.

Next we generate the edge samples by rejection sampling. We use the resistance sketch $R_e^G$ with $R_e\le R_e^G\le 4 R_e$, we set $c_e=R_e^G/16$, $C_0=\sum_e w_ec_e=\Theta(S)$. Use the probability $p_e^0=w_ec_e/C_0$. Draw $e\sim p^0$, generate a fresh query mark $\omega$, query $\widehat R_e(\omega)$ by \Cref{thm:l1-resistance-sampler}, and accept with probability
$\widehat R_e(\omega)/(32c_e)$. Otherwise repeat the sampling process independently.
The probability of accepting $(e,\omega)$ in one attempt is
\[
 p_e^0\nu_e(\omega)\frac{\widehat R_e(\omega)}{32c_e}
 =\frac{a_{e,\omega}}{32C_0}.
\]
Thus condition on the success event, the sampling distribution is $a_{e,\omega}/A=r_{e,\omega}/S$.
The expected acceptance probability is $A/(32C_0)=\Omega(1)$. The expected query time for $\widehat R_e(\omega)$ is $\widetilde O(1)$ by \Cref{thm:l1-resistance-sampler}, since we sample edge from an constant factor estimation of the resistance. For each sampled edge $(e,\omega)$, we set the edge weight by \Cref{lem:aggregate-determinant-sampling} as $w_{e,\omega}'=(w_{e,\omega}S/r_{e,\omega})(1/(t)_S)^{1/S}=( A/\widehat R_e(\omega))(1/(t)_S)^{1/S}$. Since we can not know $A$ directly, we alternatively use the estimator $\overline A$ from \Cref{lem:original-normalizer}. 

Set $\alpha=c\eps^2$ and choose 
$t=\widetilde O(\lambda^2n^2/\alpha+n^{3/2}/\sqrt\alpha)$.
Generate $t$ independent accepted samples and independently obtain
$\overline A$ from \Cref{lem:original-normalizer}. \Cref{lem:aggregate-determinant-sampling} combining with \Cref{lem:original-normalizer} gives the sampling variance at most $\exp(O(\eps^2))$. As a result, we construct the $\eps$-determinant sparsifier and finishes the proof.
\end{proof}

\section{Sparsification of Schur complements and recursive tree counting}\label{sec:schur-det}

Finally, in this section we prove \Cref{thm:main}. To approximate $\Tree(G)$, we follow the
determinant sparcification framework of \cite{DPPR} (see also
\cite{CGPSSW,LiSachdeva}): at each node of size $N$, we eliminate a constant
fraction of the vertices and recurse on two smaller children, as illustrated
in \Cref{fig:dppr-recursion}.  Concretely, one finds a
diagonally dominant set $F\subseteq V$ with
$F=\Theta(N)$, puts $T=V\setminus F$, and splits the current Laplacian
into two subgraphs.  In the first child, the vertices of $F$ are completed
to a Laplacian by adjoining a single auxiliary vertex that absorbs all
edges from $F$ into $T$; this child encodes $\det(L_{FF})$.  In the second
child, one forms the Schur complement $K=\Sc(L,T)$ on the terminal set $T$.
Block elimination and Kirchhoff's theorem then give the identity
\[
 \Tree(G)=\det(L_{FF})\,\Tree(K),
\]
so it suffices to approximate the tree counts of the two children
separately.  The completion child is an explicit graph and can be
handled by recursion.  The Schur child $K$, however, is generally dense and
expensive to compute.  We therefore replace $K$ by a determinant sparsifier,
without forming $K$ explicitly.  Because our resistance estimates satisfy
only an aggregate Pearson bound rather than a pointwise guarantee, we need to strenghen several intermediate results of \cite{DPPR} to prove the sparsification on Schur complement. 

Consider at this moment we have an explicit graph $G$ with $N$ vertices and $M$ edges. We use the following Lemma to partition $G$ into two parts.

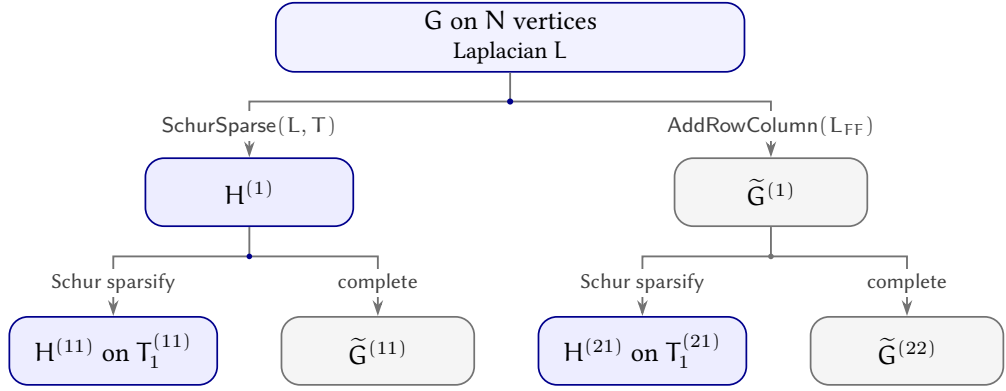
\begin{figure}
    \centering
    \makebox[\textwidth][c]{%
    \begin{tikzpicture}[
      graph/.style={
        draw=black!60,
        line width=0.65pt,
        rounded corners=7pt,
        fill=black!2,
        minimum height=9mm,
        inner xsep=9pt,
        inner ysep=3pt,
        font=\small,
        align=center
      },
      root/.style={graph, minimum width=6.2cm,
        draw=blue!55!black, fill=blue!5},
      schur/.style={graph, minimum width=2.75cm,
        draw=blue!55!black, fill=blue!7},
      completion/.style={graph, minimum width=2.55cm,
        draw=black!55, fill=black!4},
      branch/.style={draw=black!55, line width=0.7pt},
      arr/.style={branch, -{Stealth[length=2.1mm,width=1.4mm]}},
      op/.style={font=\scriptsize\sffamily, text=black!75,
        fill=white, inner sep=1.5pt, align=center}
    ]
    \node[root] (root) at (0,0)
      {$G$ on $N$ vertices\\[-1pt]{\footnotesize Laplacian $L$}};
    
    \node[schur] (L1) at (-3.45,-2.05) {$H^{(1)}$};
    \node[completion] (L2) at (3.45,-2.05) {$\widetilde G^{(1)}$};
    
    \node[schur] (L11) at (-5.25,-4.15)
      {$H^{(11)}$ on $T_1^{(11)}$};
    \node[completion] (L12) at (-1.75,-4.15) {$\widetilde G^{(11)}$};
    \node[schur] (L21) at (1.75,-4.15)
      {$H^{(21)}$ on $T_1^{(21)}$};
    \node[completion] (L22) at (5.25,-4.15) {$\widetilde G^{(22)}$};
    
    \coordinate (j0) at (0,-0.85);
    \coordinate (j1) at (-3.45,-2.90);
    \coordinate (j2) at (3.45,-2.90);
    \draw[branch] (root.south) -- (j0);
    \draw[arr] (j0) -| (L1.north);
    \draw[arr] (j0) -| (L2.north);
    \draw[branch] (L1.south) -- (j1);
    \draw[arr] (j1) -| (L11.north);
    \draw[arr] (j1) -| (L12.north);
    \draw[branch] (L2.south) -- (j2);
    \draw[arr] (j2) -| (L21.north);
    \draw[arr] (j2) -| (L22.north);
    
    \fill[blue!55!black] (j0) circle (1.15pt);
    \fill[blue!55!black] (j1) circle (1.05pt);
    \fill[black!55] (j2) circle (1.05pt);
    
    \node[op] at (-3.45,-1.18)
      {$\mathsf{SchurSparse}(L,T)$};
    \node[op] at (3.45,-1.18)
      {$\mathsf{AddRowColumn}(L_{FF})$};
    \node[op] at (-5.25,-3.25) {Schur sparsify};
    \node[op] at (-1.75,-3.25) {complete};
    \node[op] at (1.75,-3.25) {Schur sparsify};
    \node[op] at (5.25,-3.25) {complete};
    \end{tikzpicture}
    }
    \caption{Two layers of the determinant-preserving recursion
    of~\cite{DPPR}.  At each internal node, the Laplacian splits into a
    Schur-sparsification child on the terminal set~$T$ (replacing the Schur
    complement~$K$ by a sparse random graph~$H$) and a completion child on the
    eliminated set~$F$ (completing~$L_{FF}$ to a Laplacian on~$F\cup\{r\}$).}
    \label{fig:dppr-recursion}
    \end{figure}

\begin{lemma}\cite{KLP}
\label{lem:dd-set-schur}
Let $G=(V,E,w)$ be a connected weighted graph with $N$ vertices, $M$ edges, and
Laplacian $L$.  In $O(M)$ expected time one can find a $1.1$-diagonally
dominant set $F\subseteq V$ with $N/16\le |F|\le N/8$.  
\end{lemma}

Given $F$ output by \Cref{lem:dd-set-schur}, put
$T=V\setminus F$, $K=\Sc(L,T)=L_{TT}-L_{TF}L_{FF}^{-1}L_{FT}$, and
$S_T=|T|-1$.  Then $\Tree(G)=\det(L_{FF})\,\Tree(K)$. Instead of computing the Schur complement $K$ directly, we can compute the determinant sparcification of $K$ without explicitly constructing $K$. Our goal is to prove the following Lemma, which states the determinant sparcification on Schur complement.

\begin{lemma}[Aggregate-error Schur sparsification]
    \label{lem:aggregate-schur-sparsification}\label{lem:marked-schur-sparsification}
    Let $F$ selected by \Cref{lem:dd-set-schur}, $T=V\setminus F, K=\Sc(L,T)$.
    For every $0<\alpha<1/10$, there is a procedure that outputs a random graph
    $H$ on $T$ with
    $\widetilde O(N^{3/2}\alpha^{-1/2}+\lambda^2 N^2\alpha^{-1})$ edges.  Its
    expected running time is
    $\widetilde O(M+N\lambda^{-2}+N^{3/2}\alpha^{-1/2}
    +\lambda^2 N^2\alpha^{-1})$.  Moreover, $H$ satisfies
    \[
     e^{-\alpha}\Tree(K)\le \mathbb E[\Tree(H)]\le e^{\alpha}\Tree(K),
     \qquad
     \frac{\mathbb E[\Tree(H)^2]}{(\mathbb E[\Tree(H)])^2}\le e^{\alpha}.
    \]
    \end{lemma}

We first use \Cref{lem:aggregate-schur-sparsification} to prove \Cref{thm:main}, and then we prove \Cref{lem:aggregate-schur-sparsification}. As discussed above, we follow the determinant sparsification strategy of \cite{DPPR}.

\begin{proof}[Proof of \Cref{thm:main}]
We use the recursion of \cite[Section~6]{DPPR}, with our
\Cref{lem:aggregate-schur-sparsification} in place of its Schur sparsifier. See \Cref{fig:dppr-recursion} as an illustration of this framework.
Consider at the current level of recursion, we have a graph $G_x=(V_x,E_x,w_x)$, let $L=L_x$ be the current Laplacian, let $N=N_x$ be its number of
vertices.  Stop at
$N\le O(\lambda^{-1})$ and evaluate the reduced
determinant directly, so every internal node satisfies $1/N<\lambda$.
Choose the $1.1$-diagonally dominant set $F$ from
\Cref{lem:dd-set-schur}, set $T=V_x\setminus F$, and write
$K=\Sc(L,T)$.  One child is the random Laplacian $H$ returned by
\Cref{lem:aggregate-schur-sparsification}; the other is obtained by completing
$L[F,F]$ to a Laplacian with one additional vertex.  
Kirchhoff's theorem and block elimination gives us that
\[
 \Tree(L)=\det L[F,F]\,\Tree(K),
 \qquad
 \Tree(H)\det L[F,F]
   =\Tree(L)Y,
 \quad
 Y:=\frac{\Tree(H)}{\Tree(K)}.
\]

It is easy to see that the recursion depth is $D=O(\log n)$.
We next verify the accumulation of the random error.  Put $r_0=n-1$,
let $\delta_0=c_0\eps^2/D$ for a sufficiently small absolute constant
$c_0>0$. We give the current recursion graph $G_x$ an error budget
$\alpha_x=\delta_0 (N_x-1)/(n-1)$. Since the recursion depth is $D$, summing over all $x$ gives us $\sum_x\alpha_x\le D\delta_0$.

Let $G_w$ be the parant graph of $G_x$ in the recursion. So conditioning on the parent graph $G_w$,
the expectation and variance analysis by \Cref{lem:aggregate-schur-sparsification} gives
$\mathbb E[Y_x\mid G_w]=\exp(\pm\alpha_x)$ and
$\mathbb E[Y_x^2\mid G_w]\le \exp (3\alpha_x)$.
The final output ratio $W=\prod_xY_x$ therefore satisfies
\[
 \mathbb E(W-1)^2
 \le \exp(3D\delta_0)-2\exp(-D\delta_0)+1=O(D\delta_0).
\]
By setting $\delta_0=c_0\eps^2/D$, this gives an $\eps$-approximation
to $\Tree(G)$ with probability $2/3$.

It remains to check the total runtime.  Consider all the graphs $\{G_x\}$ on same recursion level, we have that 
$\sum_{\text{same level }x}N_x=O(n)$.  Substituting
$\alpha_x=\delta_0(N_x-1)/(n-1)$ into
\Cref{lem:aggregate-schur-sparsification} and summing over the same level graphs gives us
\[
 \sum_{\text{same level }x}\left(N_x^{3/2}\alpha_x^{-1/2}
       +\lambda^2 N_x^2\alpha_x^{-1}
       +N_x\lambda^{-2}\right)
 =\widetilde O\!\left(n^{3/2}\delta_0^{-1/2}
       +\lambda^2 n^2\delta_0^{-1}+n\lambda^{-2}\right).
\]

The recursion has $D=O(\log n)$ total levels. So the total expected runtime is
$\widetilde O(m+n\lambda^{-2}+n^{3/2}\eps^{-1}
+\lambda^2 n^2\eps^{-2})$. This proves the theorem by setting
$\lambda=n^{-1/4}\eps^{1/2}$.
\end{proof}

Next we prove \Cref{lem:aggregate-schur-sparsification}. Our overall goal is the three objects: (i) sampling edge of the Schur complement $K$ proportional to the estimation of leverage score on $K$, without explicitly constructing $K$;  (ii) showing that the Pearson aggregate error between the real leverage score and the estimated leverage score on Schur complement is $\widetilde O(\lambda^2N)$; (iii) showing that the sample complexity of an edge is $\widetilde O(1)$ in expectation. If these three objects are satisfied, then we can then prove \Cref{lem:aggregate-schur-sparsification} by directly using the sample variance bound of \Cref{lem:aggregate-determinant-sampling}.

We begin with the random walk representation of the Schur complement by \cite{DPPR}. They show that the edges of Schur can be viewd as the random walks in the original graph.

\begin{lemma}\cite[Section~5]{DPPR}
\label{lem:schur-path-representation}
Write $L_{FF}=D_{FF}-A_{FF}$.  Let $\mathcal P_{T,F}$ contain one
labelled undirected path
$p=(u_0,u_1,\ldots,u_k)$ for every finite walk whose distinct endpoints
$u_0,u_k$ lie in $T$ and whose internal vertices lie in $F$. Give $p$ weight defined as
\[
 w_p=\frac{\prod_{i=0}^{k-1}w_{u_i u_{i+1}}}
 {\prod_{i=1}^{k-1}d_{u_i}}.
\]
Then the Schur complement is the Laplacian of this labelled
multigraph:
\[
 K=\sum_{p\in\mathcal P_{T,F}}w_p
 b_{u_0,u_k}b_{u_0,u_k}^{T}.
\]
Consequently, a labelled edge of $K$ may be viewed as a path that starts in
$T$, travels only through $F$, and returns to $T$.  Its leverage score is
$\ell_p=w_pR_K(u_0,u_k)$, and
$\sum_{p\in\mathcal P_{T,F}}\ell_p=|T|-1$.
\end{lemma}

By \Cref{lem:schur-path-representation}, we have the following overall strategy for sampling edge on the Schur complement $K$ proportional to the estimation of leverage score: we first sample a labelled path $p\in \mathcal P_{T,F}$ via random walk, assuming $p=(u_0,u_1,\ldots,u_k)$ with $u_0,u_k\in T$. Then we compute the estimation of resistance $\widehat R^G(u_0,u_k)$ on original graph $G$. By \Cref{prop:schur-original-resistance}, $R^G_{u_0,u_k}=R^{\Sc(L,T)}_{u_0,u_k}$, so this can be also viewed as the estimation of resistance on the Schur complement. Then we use rejection sampling to guarantee that the sampled edge is proportional to the estimation of leverage score on the Schur complement.

We first focus on showing that the aggregate Pearson error of the estimated resistance on the Schur complement is $\widetilde O(\lambda^2N)$. Recall that in \Cref{sec:det-sparsifier}, we conceptually split graph $G$ with virtual parallel edges. For each edge $e$, let the query output of resistant estimator be marked as $\widehat R_e(\omega)$, with occur probability $\nu_e(\omega)$. We split $e$ into parallel copies $(e,\omega)$ with weights $w_{e,\omega}=w_e\nu_e(\omega)$.  The aggregate Pearson error bound holds for this virtual graph by \Cref{thm:l1-resistance-sampler}, and next we need to show this also holds for Schur complement $K$. 

For Schur complement, we use the similar construction. Let $1/N<\lambda\le1/4$ and $S_T=|T|-1=\Theta(N)$.
Consider a path $p\in\mathcal P_{T,F}$ with endpoints $s,t\in T$. Assume the resistance query output $\widehat R_{s,t}(\omega)$ with probability $\nu_{s,t}(\omega)$. For the virtual copy $(p,\omega)$, set
\[
 a_{p,\omega}=w_p\nu_{st}(\omega)a_{st}(\omega),\qquad
 \ell_{p,\omega}=w_p\nu_{st}(\omega)R_K(s,t).
\]
Also set
\[
A=\sum_{p,\omega}a_{p,\omega}, \quad \text{and} \quad
r_{p,\omega}=S_Ta_{p,\omega}/A. 
\]

Clearly, it holds that $\sum_{p,\omega}r_{p,\omega}=S_T$ and $r_{p,\omega}\ge\ell_{p,\omega}/C$ by \Cref{thm:l1-resistance-sampler}. We next show the aggregate Pearson error of the resistance estimation preserves for Schur complement $K$.

\begin{lemma}
\label{lem:schur-path-scores}
It holds that
\[
 \sum_{p,\omega}\frac{(r_{p,\omega}-\ell_{p,\omega})^2}{r_{p,\omega}}
 =O(\lambda^2N).
\]

\end{lemma}

\begin{proof}
Fix a successful preprocessing state of \Cref{thm:l1-resistance-sampler}
and write $a_{st}=\widehat R_{s,t}$ for its query output.
Throughout this proof, $B_{s,t}$ and $E_{s,t}$ are defined on the parent graph.
By the variance bound of \Cref{lem:original-correction-variance} and the sample counts chosen in
the proof of \Cref{thm:l1-resistance-sampler}, we have
\[
 \mathbb E_{\mathrm{query}}
       \frac{(a_{st}-R_G(s,t))^2}{a_{st}}
 \le C\lambda B_{s,t}.
\]
For any vector $f\in\mathbb R^V$, by the definition of Laplacian matrix, we have
\[
 \sum_pw_p(f(s)-f(t))^2=f_T^TKf_T
 \le f^TLf=\sum_{\{s,t\}\in E(G)}w_{st}(f(s)-f(t))^2.
\]
More generally, we consider vector $u_s\in \mathbb R^d$ for every $s\in V$. Fix coordinate $i$, we define $f_i(s)=(u_s)(i)$. By the argument before, we have
\[
\sum_p w_p(u_s(i)-u_t(i))^2\le \sum_{\{s,t\}\in E(G)}w_{st}(u_s(i)-u_t(i))^2.
\]
Summing over all coordinates $i$, we have
\[
\sum_p w_p \|u_s-u_t\|_2^2\le \sum_{\{s,t\}\in E(G)}w_{st}\|u_s-u_t\|_2^2.
\]
Now we notice that
\begin{align*}
  B_{s,t}&=
E_{s,t}+\frac{\lambda}{\Lambda} \|\Psi_s-\Psi_t\|_2^2+\lambda R_G(s,t)\\
&=\|L^{\dagger/2}(\widetilde q_s-\widetilde q_t)\|_2^2
+\frac{\lambda}{\Lambda} \|\Psi_s-\Psi_t\|_2^2+
\lambda \|L^{\dagger/2}(e_s-e_t)\|_2^2
\end{align*}
is the sum of squared distance, therefore
\[
 \sum_pw_pB_{s,t}
 \le\sum_{\{s,t\}\in E(G)}w_{st}B_{s,t}=O(\lambda N).
\]
Since $R_K(s,t)=R_G(s,t)$ by \Cref{prop:schur-original-resistance}, we have
\[
 \sum_{p,\omega}
       \frac{(a_{p,\omega}-\ell_{p,\omega})^2}{a_{p,\omega}}
 =\sum_pw_p\mathbb E_{\mathrm{query}}
       \frac{(a_{st}-R_K(s,t))^2}{a_{st}}
 \le C\lambda\sum_pw_pB_{s,t}=O(\lambda^2N).
\]
This finishes the proof.
\end{proof}

The next goal is to show that we can sample Schur path proportional to estimated leverage score in $\widetilde O(1)$ time in expectation. The following lemma records the Schur-path proposal of
\cite{DPPR}, which generates a coarse sample of a labelled path of the implicit Schur
complement. Our algorithm for sampling Schur path is shown in \Cref{alg:schur-two-stage-rejection}. We use \Cref{lem:path-proposal} as a subroutine.

\begin{lemma}[\cite{DPPR}, Lemma 5.5]
\label{lem:path-proposal}
Let $F$ be a $1.1$-diagonally dominant set, $T=V\setminus F$,
$K=\Sc(L,T)$, and $|T|-1=\Theta(N)$. After $\widetilde O(M)$ preprocessing,
there is a sampler that returns a labelled
Schur path $p=(u_0,u_1,\ldots,u_k)\in\mathcal P_{T,F}$ together with its
unconditioned proposal probability $p_p^0$, such that
\[
 \ell_p:=w_pR_K(u_0,u_k)\le \rho (|T|-1)p_p^0
\]
for some absolute constant $\rho=O(1)$. The expected path length and sampling time
are $\widetilde O(1)$.
\end{lemma}

\begin{algorithm}
\caption{sampling the Schur path}
\label{alg:schur-two-stage-rejection}
\KwIn{A preprocessed graph $G=(V,E,w)$ with the fixed sets $F,T$}
\KwOut{A labelled Schur path $p$}
$S\leftarrow |T|-1$; let $\rho$ be the constant in \Cref{lem:path-proposal}\;
\While{true}{
  Draw a path $p$ and its probability $p_p^0$ by \Cref{lem:path-proposal}\;
    Let $s,t$ be the endpoints of $p$ and compute its weight $w_p$\;
    Query $c_{st}\leftarrow R^G_{s,t}/16$ by resistance sketch (\Cref{lem:res-sketch})\;
    Draw independent $U_1,U_2\sim\operatorname{Unif}[0,1]$\;
    \tcp{Stage 1: cheap rejection.}
    \If{$U_1\le w_pc_{st}/(\rho S p_p^0)$}{
      \tcp{Stage 2: query the resistance and reject again.}
      Generate a query mark $\omega$ and obtain $a_{st}(\omega)\leftarrow\widehat R_{s,t}(\omega)$ by \Cref{thm:l1-resistance-sampler}\;
      \If{$U_2\le a_{st}(\omega)/(32c_{st})$}{
        \Return{$p$}\;
      }
    }
}
\end{algorithm}

\begin{lemma}
\label{lem:schur-two-stage-rejection}
After preprocessing, \Cref{alg:schur-two-stage-rejection} samples a labelled path $p$ with its mark $\omega$ with probability proportional to $a_{p,\omega}$ in $\widetilde O(1)$ expected time.
\end{lemma}

\begin{proof}
We write
$R_{st}=R_G(s,t)=R_K(s,t)$ and $S=|T|-1$.
The resistance sketch and \Cref{thm:l1-resistance-sampler} guarantees
\[
 R_{st}/16\le c_{st}\le R_{st}/4,
 \qquad c_{st}\le a_{st}(\omega)\le32c_{st}.
\]
By \Cref{lem:path-proposal}, $w_pR_{st}\le\rho S p_p^0$.
Thus both acceptance probabilities are at most one.
The probability of
proposing and accepting a path $(p,\omega)$ is
\[
 p_p^0w_p\frac{c_{st}}{\rho S p_p^0}\nu_{st}(\omega)
       \frac{a_{st}(\omega)}{32c_{st}}
 =\frac{a_{p,\omega}}{32\rho S}.
\]

Conditioning on the success event, the proposed Schur path $p$ has probability proportional to $a_{p,\omega}$. Since $\sum_pw_pR_{st}=S$ and
$R_{st}/16\le a_{st}(\omega)\le8R_{st}$, we have
$S/16\le A\le8S$.
So the total acceptance probability is therefore
$A/(32\rho S)=\Theta(1)$.

For the running time, we notice that conditioning on the success event of the first stage rejection sampling, the proposed Schur path $p$ is porportional to the constant factor approximation of leverage scores. Therefore, the expected time for querying $\widehat R_{s,t}$ is $\widetilde O(1)$ by \Cref{thm:l1-resistance-sampler}.
\end{proof}

Therefore we prove the three ingredients for the proof of \Cref{lem:aggregate-schur-sparsification}. There is one additional wrinkle:  we need to assign the weight to each sampled path $p$ by \Cref{lem:aggregate-determinant-sampling}. The sampled edge weights require the normalizer $A$. However, we can not know $A$ directly. Similar to \Cref{lem:original-normalizer}, we use the following lemma to efficiently estimate $A$. We defer its proof in \Cref{sec:missing-proofs}.

\begin{lemma}
\label{lem:schur-path-normalizer}\label{lem:marked-sampling-normalizer}
For every $0<\alpha<1/10$, in
$\widetilde O(\lambda^2 N^2\alpha^{-1}+N\alpha^{-1/2})$ expected
time, we can construct an estimate $\overline A\ge (|T|-1)/32$,
independently of the determinant samples, such that for
$Z=(\overline A/A)^{|T|-1}$,
\[
 e^{-\alpha}\le\mathbb E [Z]\le e^{\alpha},
 \qquad \frac{\mathbb E [Z^2]}{(\mathbb E [Z])^2}\le e^{\alpha}.
\]
\end{lemma}

We put things together to prove \Cref{lem:aggregate-schur-sparsification}.

\begin{proof}[Proof of \Cref{lem:aggregate-schur-sparsification}]
Apply \Cref{lem:schur-path-scores} to obtain normalized path-copy scores with
Pearson error $O(\lambda^2 N)$.  Invoke
\Cref{alg:schur-two-stage-rejection} to sample from $r_{p,\omega}/S_T$, and independently
invoke \Cref{lem:schur-path-normalizer} to estimate the normalizer.  Take
$t=\widetilde O(N^{3/2}\alpha^{-1/2}+\lambda^2 N^2\alpha^{-1})$ samples, \Cref{lem:aggregate-determinant-sampling,lem:marked-sampling-normalizer} guarantees that the variance is at most $\exp (O(\alpha))$. Each sample costs $\widetilde O(1)$ expected time by \Cref{lem:schur-two-stage-rejection}, together with the preprocessing runtime $\widetilde O(M+N\lambda^{-2})$, the total expected time is $\widetilde O(M+N\lambda^{-2}+N^{3/2}\alpha^{-1/2}+\lambda^2 N^2\alpha^{-1})$.
\end{proof}

\section{Missing Proofs.}\label{sec:missing-proofs}

In this section, we provide missing proofs for several technical lemmas.

\begin{proof}[Proof of \Cref{lem:aggregate-determinant-sampling}]
  We first prove $\mathbb{E}[\Tree(H)]=\Tree(G)$. Let $X_e$ be the number of copies of $e$ among the $t$ samples.  Each copy
  of $e$ has weight $w'_e$, so the total weight of $e$ in $H$ is $X_e w'_e$.
  Therefore
  \[
   \Tree(H)=\sum_{T\text{ spanning}}\ \prod_{e\in T}(X_e w'_e),
  \]
  where the sum ranges over all spanning trees of $G$.
  By linearity of expectation,
  \[
   \mathbb E[\Tree(H)]
   =\sum_{T\text{ spanning}}
     \Bigl(\prod_{e\in T}w'_e\Bigr)
     \mathbb E\Bigl[\prod_{e\in T}X_e\Bigr].
  \]
  Fix a spanning tree $T$.  Expanding $\prod_{e\in T}X_e$ amounts to choosing,
  for each edge of $T$, one of the $t$ sample positions. Only choices with $S$
  pairwise distinct positions contribute, so there are $(t)_S$ such ordered
  choices.  Each choice has probability contribution $\prod_{e\in T}p_e$, so
  \[
   \mathbb E\Bigl[\prod_{e\in T}X_e\Bigr]
   =(t)_S\prod_{e\in T}p_e.
  \]
  By our definition of $w'_e$, we have that
  \[
   \mathbb E\Bigl[\prod_{e\in T}(X_e w'_e)\Bigr]
   =\Bigl(\prod_{e\in T}w_e\Bigr)
     \frac{S^S}{(t)_S\prod_{e\in T}r_e}
     \cdot (t)_S \prod_{e\in T}\frac{r_e}{S}
   =\prod_{e\in T}w_e.
  \]
  Summing over all spanning trees $T$ yields $\mathbb E[\Tree(H)]=\Tree(G)$.
  
  For the second moment, let $\mu$ be the weighted spanning-tree probability
  measure of $G$, so $\mu(T)\propto\prod_{e\in T}w_e$ and
  $\mu(F\subseteq T):=\sum_{T:\,F\subseteq T}\mu(T)$.
  Squaring the display above and dividing by $\Tree(G)^2$ gives
  \[
   \frac{\mathbb E[\Tree(H)^2]}{\Tree(G)^2}
   =\sum_{T_1,T_2}
     \mathbb E\Bigl[\prod_{e\in T_1}(X_e w'_e)
                    \prod_{e\in T_2}(X_e w'_e)\Bigr]
     \Big/\Tree(G)^2.
  \]
  Fix an ordered pair of spanning trees $(T_1,T_2)$.  On every edge
  $e\in T_1\cap T_2$ we have a factor $X_e^2$, which we expand as
  $X_e^2=(X_e)_2+(X_e)_1$, where $(X_e)_1=X_e$ is linear term and
  $(X_e)_2=X_e(X_e-1)$ is quadratic term, counting the number of ordered
  pairs of distinct sample positions among the $t$ draws that both land on
  $e$.
  Choosing the linear term on a subset $F\subseteq T_1\cap T_2$ and the
  quadratic term on the remaining common edges yields a sum over such $F$.
  Edges in exactly one of $T_1,T_2$ contribute a single factor $(X_e)_1$.
  Inserting this expansion into the previous equation gives
  \begin{align*}
   \frac{\mathbb E[\Tree(H)^2]}{\Tree(G)^2}
   &=\sum_{T_1,T_2}\ \sum_{F\subseteq T_1\cap T_2}
     \frac{1}{\Tree(G)^2}\,
     \mathbb E\Bigl[
     \prod_{e\in F}(X_e)_1(w'_e)^2
     \prod_{e\in(T_1\cap T_2)\setminus F}(X_e)_2(w'_e)^2 \\
   &\qquad\qquad\times
     \prod_{e\in T_1\setminus T_2}(X_e)_1 w'_e
     \prod_{e\in T_2\setminus T_1}(X_e)_1 w'_e
     \Bigr].
  \end{align*}
  Thus each term in the expansion is a product of falling factorials
  $(X_e)_{j_e}$ with $j_e\in\{1,2\}$.  The multinomial identity
  \[
   \mathbb E\Bigl[\prod_e(X_e)_{j_e}\Bigr]
   =(t)_{\sum_e j_e}\prod_e p_e^{j_e}
  \]
  applies to every such term.  For a fixed $F$, the total exponent is
  $\sum_e j_e=2S-|F|$, because the two trees contribute $2S$ linear slots
  and each edge of $F$ replaces a quadratic slot by a linear one.  The
  sampled-copy weights contribute a factor $(t)_S^{-2}$ from the two trees,
  so the sampling expectation carries the prefactor
  $(t)_{2S-|F|}/(t)_S^2$.  Moreover, every $e\in F$ leaves one extra factor
  $p_e^{-1}$ after the weight normalization, whereas the tree weights over
  $(T_1,T_2)$ sum to $\mu(F\subseteq T)^2$.  Collecting these contributions
  over all pairs $(T_1,T_2)$ and all subsets $F\subseteq T_1\cap T_2$ gives
  \[
   \frac{\mathbb E[\Tree(H)^2]}{\Tree(G)^2}
   =\sum_{\substack{F\subseteq E\\ |F|\le S}}
     \frac{(t)_{2S-|F|}}{(t)_S^2}
     \left(\prod_{e\in F}\frac1{p_e}\right)
     \mu(F\subseteq T)^2.
  \]
  Weighted spanning trees are negatively associated~\cite{BurtonPemantle,DPPR}, so
  $\mu(F\subseteq T)\le\prod_{e\in F}\ell_e$. Substituting into the collision expansion and grouping by $k=|F|$ gives
  \[
   \frac{\mathbb E[\Tree(H)^2]}{\Tree(G)^2}
   \le\sum_{k=0}^S
     \frac{(t)_{2S-k}}{(t)_S^2}
     \sum_{\substack{F\subseteq E\\ |F|=k}}
     \prod_{e\in F}\frac{\ell_e^2}{p_e}.
  \]
  The inner sum is the elementary symmetric polynomial of degree $k$ in the
  numbers $\ell_e^2/p_e$, By AM-GM inequality, it holds that
  \[
   \sum_{\substack{F\subseteq E\\ |F|=k}}
   \prod_{e\in F}\frac{\ell_e^2}{p_e}
   \le\frac1{k!}\left(\sum_e\frac{\ell_e^2}{p_e}\right)^k
  \]
  As a result, we have that
  \[
   \frac{\mathbb E[\Tree(H)^2]}{\Tree(G)^2}
   \le\sum_{k=0}^S
   \frac{(t)_{2S-k}}{(t)_S^2k!}
   \left(\sum_e\frac{\ell_e^2}{p_e}\right)^k.
  \]
  The Pearson identity is
  $\sum_e\ell_e^2/p_e=S\sum_e\ell_e^2/r_e
  =S(S+\Xi(r\Vert\ell))$.
  Put $d=t-2S+1$.  Since $t\ge4S$,
  $(t)_{2S-k}/(t)_{2S}\le d^{-k}$. So we have
  \[
      \frac{\mathbb E[\Tree(H)^2]}{\Tree(G)^2}\le
      \frac{(t)_{2S}}{(t)_S^2}\sum_{k=0}^S{\frac{1}{k!}\left(\frac{S(S+\Xi(r\Vert\ell)))}{d}\right)^k}\le
      \frac{(t)_{2S}}{(t)_S^2}\exp\left(\frac{S(S+\Xi(r\Vert\ell))}{d}\right)
  \]

  Next, by the fact that
  \[
  \log\left(\frac{(t)_{2S}}{(t)_S^2}\right)=\sum_{j=0}^{S-1}\log\left(1-\frac{S}{t-j}\right)\le
  -\sum_{j=0}^{S-1}\frac{S}{t-j}\le-S^2/t,
  \]
  we have
  \[
   \log\frac{\mathbb E[\Tree(H)^2]}{\Tree(G)^2}
   \le \frac{S(S+\Xi(r\Vert\ell))}{d}-\frac{S^2}{t}=O\left(\frac{S\Xi(r\Vert\ell)}{t}+\frac{S^3}{t^2}\right).
  \]
  This proves the second-moment bound.
  \end{proof}

  \begin{proof}[Proof of \Cref{lem:original-normalizer}]
    Write $\widehat R_e^\lambda=b_e^T(\widetilde z_s-\widetilde z_t)$ and
    $W_e=\widehat R^{(0)}_{s,t}$, where $e=\{s,t\}$.
    Set $c_e=R^G_e/16$ for the resistance sketch $R^G_e$, so that
    $R_e/16\le c_e\le R_e/4$ and
    $\widehat R_e=\min\{32c_e,\max\{c_e,W_e\}\}$.
    Using the estimates $\widetilde B_e=\Theta(B_e)$ constructed in the
    proof of \Cref{thm:l1-resistance-sampler}, scan the original edges and
    build alias tables for
    \[
     p_e^0=\frac{w_ec_e}{C_0},\quad C_0=\sum_e w_ec_e=\Theta(S),
     \qquad
     p_E(e)=\frac{w_e\widetilde B_e}{U},\quad
     U=\sum_e w_e\widetilde B_e=O(\lambda n).
    \]
    This costs $\widetilde O(m)$ time, by
    \Cref{lem:original-budget-mass}. The same proof of
    \Cref{thm:l1-resistance-sampler} gives
    \[
     \sum_e w_e\frac{\mathbb E(W_e-R_e)^2}{R_e}=O(\lambda^2n).
    \]
    The trace identity
    $\sum_e w_e\widehat R_e^\lambda=\operatorname{Tr}(L[\widetilde z_s]_s)
    =n-\sum_s\widetilde q_s(s)$ gives the decomposition
    \[
     A=n-\sum_s\widetilde q_s(s)+C_{\rm raw}+C_{\rm clip},\qquad
     C_{\rm raw}=\sum_e w_e\mathbb E(W_e-\widehat R_e^\lambda),\quad
     C_{\rm clip}=\sum_e w_e\mathbb E(\widehat R_e-W_e).
    \]
    We estimate $C_{\rm raw}$ and $C_{\rm clip}$ separately to obtain the estimation of $A$. 
    
    \textbf{Estimating $C_{\rm raw}$:} For $C_{\rm raw}$, draw $e\sim p_E$, choose
    $J$ uniformly from $\{1,\ldots,d_e\}$, and independently generate one sample $X_1$ from \Cref{alg:original-slack-query}. Put the estimator $Y_{\rm raw}$ as
    \[
     D_e=(\mathcal T_{J,s}-\mathcal T_{J,t})^2+X_1,
     \qquad Y_{\rm raw}=\frac{U D_e}{\widetilde B_e}.
    \]
    Conditional on the index, $\mathbb ED_e=\mathbb E(W_e-\widehat R_e^\lambda)$. Thus $\mathbb EY_{\rm raw}=C_{\rm raw}$,
    and one trial takes $\widetilde O(1)$ time, without performing $d_e$
    queries. \Cref{lem:original-correction-variance} give
    $(\mathcal T_{j,s}-\mathcal T_{j,t})^2\le\widetilde O(B_e)$ with probability $1-n^{-O(1)}$. The pointwise bound in the proof of
    \Cref{lem:original-slack-unbiased} gives
    $|X_1|\le C\lambda(\|\Psi_s-\Psi_t\|_2^2+R_e)
    \le \widetilde O(B_e)$. Consequently
    \[
     |Y_{\rm raw}|=\widetilde O(\lambda n),\qquad
     \operatorname{Var}Y_{\rm raw}=\widetilde O(\lambda^2n^2).
    \]
    
    \textbf{Estimating $C_{\rm clip}$:} For $C_{\rm clip}$, independently draw $e\sim p^0$, execute the
    complete query to obtain $W_e,\widehat R_e$, and return the estimator
    \[
     Y_{\rm clip}=\frac{C_0}{c_e}(\widehat R_e-W_e).
    \]
    This is unbiased. The expected time cost is
    $\widetilde O(1)$ by \Cref{thm:l1-resistance-sampler}, since we query from constant approximation of leverage score.
    The preceding pointwise bounds also give $|W_e|=\widetilde O(R_e)$.
    Therefore $|Y_{\rm clip}|=\widetilde O(n)$ and

\begin{align*}
    \mathbb EY_{\rm clip}^2&=\sum_e \frac{w_ec_e}{C_0}\frac{C_0^2}{c_e^2}(\mathbb E(\widehat R_e-W_e)^2)\\
    &=C_0\sum_e \frac{w_e}{c_e}(\mathbb E(\widehat R_e-W_e)^2)\\
    &\le 16C_0\sum_e w_e\frac{\mathbb E(\widehat R_e-W_e)^2}{R_e}=O(\lambda^2n^2).
\end{align*}

    Average $h=\widetilde O(\lambda^2n^2/\beta+n/\sqrt\beta)$ independent
    trials of each type, using independent randomness for the two types,
    and add the known term $n-\sum_s\widetilde q_s(s)$ to obtain
    $\widehat A$. For $0<\beta<1/10$, bounded-variable Bernstein \cite{BoucheronLugosiMassart} gives
    \[
     \mathbb E\widehat A=A,\qquad
     \operatorname{Var}(\widehat A-A)=O(\beta),\qquad
     \log\mathbb E e^{\theta(\widehat A-A)}=O(\theta^2\beta).
    \]
     Set $\overline A=\max\{S/32,\widehat A\}$.

    It remains to pass from $\widehat A$ to $Z=(\overline A/A)^S$.
    Since $S/16\le A\le8S$, Chebyshev's inequality gives
    \[
     \Pr[\widehat A<S/32]
     \le\frac{\operatorname{Var}(\widehat A)}{(A-S/32)^2}
     =O(\beta/S^2).
    \]
    For $j\in\{1,2\}$, on the event $\widehat A\ge S/32$,
    the inequality $\log x\le x-1$ gives
    \[
     Z^j=\left(\frac{\widehat A}{A}\right)^{jS}
     \le\exp\!\left(\frac{jS}{A}(\widehat A-A)\right).
    \]
    The coefficient $jS/A\le32$, so the preceding moment bound applies.
    On the complementary event, $Z^j=(S/(32A))^{jS}\le2^{-jS}$.
    Consequently,
    \[
     \mathbb EZ^j
     \le e^{O(\beta)}+2^{-jS}O(\beta/S^2)
     \le e^{O(\beta)},\qquad j=1,2.
    \]
    On the other hand, $\overline A\ge\widehat A$ implies
    $\mathbb E\overline A\ge A$, and Jensen's inequality yields
    \(
     \mathbb EZ\ge\left(\frac{\mathbb E\overline A}{A}\right)^S\ge1.
    \)
    Taking $\beta$ a sufficiently small constant multiple of $\alpha$,
    we obtain
    \[
     e^{-\alpha}\le\mathbb EZ\le e^\alpha,\qquad
     \frac{\mathbb EZ^2}{(\mathbb EZ)^2}\le e^\alpha.
    \]
    This proves the lemma.
    \end{proof}

\begin{proof}[Proof of \Cref{lem:marked-sampling-normalizer}]
  Put $S=|T|-1$, $W_{st}=\widehat R^{(0)}_{s,t}$,
  $a_{st}=\widehat R_{s,t}$, and
  $\widehat R^\lambda_{st}=b_{st}^T(\widetilde z_s-\widetilde z_t)$. They are the queries on parent graphs.
  
  Let $\widetilde q_s=e_s-L\widetilde z_s$ and
  $P=-L_{FF}^{-1}L_{FT}$.
  Block elimination gives
  \[
   K\widetilde z_{s,T}=e_s-\bar q_{s,T},
   \qquad
   \bar q_{s,T}=\widetilde q_{s,T}+P^T\widetilde q_{s,F}.
  \]
  Writing $\Theta=\sum_{s\in T}\bar q_{s,T}(s)$, the trace calculation
  of \Cref{lem:original-normalizer} therefore becomes
  \[
   \begin{split}
   A&=|T|-\Theta+C_{\rm raw}+C_{\rm clip},\\
   C_{\rm raw}&=\sum_pw_p\mathbb E(W_{st}-\widehat R^\lambda_{st}),
   \qquad
   C_{\rm clip}=\sum_pw_p\mathbb E(a_{st}-W_{st}).
   \end{split}
  \]
  
  \textbf{Estimating $\Theta$.}
  In the original graph, the trace $\sum_s q_s(s)$ can be directly computed. However, here we have
  \[
   \Theta=\sum_{s\in T}\widetilde q_s(s)
         +\sum_{x\in F,s\in T}P_{xs}\widetilde q_s(x).
  \]
  The first sum is known, but directly compute the second term is expensive. So we need to estimate the second term. We choose $X$ uniformly from $F$
  and run a random walk until its first terminal $S_X\in T$.
  The distribution of $S_X$ is $P_{X,\cdot}$ by the definition of $P$.
  Next we use the residual point estimate $Y_{S_X}(X)$ from
  \Cref{alg:original-slack-query}. Specifically, for a fixed sample, let $x=X$, $s=S_X$. If $x\in A_s$, we directly return $Y_s(x)=\widetilde q_s(x)$; otherwise, we use the heavy/light estimate as in \Cref{alg:original-slack-query}:
  \[
  Y_s(x)=\sum_{a\in A_s,w_{xa}\widetilde z_s(a)>\lambda^2} w_{xa} \widetilde z_s(a)+|A_s|w_{xa}\widetilde z_s(a)\one_{w_{xa}\widetilde z_s(a)\le\lambda^2}(x),
  \]
  where $a\in A_s$ is chosen uniformly at random. We use $\sum_{s\in T}\widetilde q_s(s)+|F|Y_{S_x}(X)$ as the estimation of $\Theta$. By \Cref{lem:original-energy-unbiased}, We have $\mathbb EY_s(x)=\widetilde q_s(x)$ and
  $0\le Y_s(x)\le3\lambda$ for $x\ne s$.
  Thus $|F|Y_{S_X}(X)$ is unbiased for the second sum, and has variance $O(\lambda^2N^2)$. The expected time for simulating random walk is $\widetilde O(1)$ by our choice of $F$ in \Cref{lem:dd-set-schur}.
  
  \textbf{Estimating $C_{\rm raw}$.}
  Choose a seed edge with
  \[
   U=\sum_{e\in E(G)}w_e\widetilde B_e=O(\lambda N),
   \qquad \Pr[e]=w_e\widetilde B_e/U.
  \]
  We extend both endpoints of the seed edge by
  independent random walks stopped on $T$.
  If their terminal endpoints coincide, return zero.
  For a labelled path $p$, summing over its possible seed
  occurrences gives its proposed probability $q_E(p)$ with
  \[
   q_E(p)=\frac{w_pD_p}{U},
   \qquad D_p=\sum_{e\text{ occurrence in }p}\widetilde B_e,
  \]
  where repeated occurrences are counted separately.
  Generate the single sample $D_{st}$ in the proof of
  \Cref{lem:original-normalizer} (i.e., the estimator of $C_{\rm raw}$ on the original graph), and return the estimation of $C_{\rm raw}$:
  \[
   Y_{\rm raw}=\frac{UD_{st}}{D_p}.
  \]
  By \Cref{lem:original-normalizer}, we have $\mathbb ED_{st}=\mathbb E(W_{st}-\widehat R^\lambda_{st})$,
  so $\mathbb EY_{\rm raw}=C_{\rm raw}$. Moreover, $|D_{st}|\le \widetilde O(B_{st})$. By
  \[
   B_{st}\le |p|\sum_{e\text{ occurrence in }p}B_e
            \le O(|p|D_p),
  \]
  We have $|Y_{\rm raw}|\le \widetilde O(U|p|)=\widetilde O(\lambda N)$ with high probability, since the path $p$ has length $\widetilde O(1)$ with high probability. So the variance of $Y_{\rm raw}$ is $\widetilde O(\lambda^2N^2)$.

  \textbf{Estimating $C_{\rm clip}$.}
  Draw $p$ from \Cref{lem:path-proposal}.
  Apply only the first-stage rejection gate of
  \Cref{alg:schur-two-stage-rejection}.
  Generate $W_{st}$ and $a_{st}$ from the resistance query by \Cref{thm:l1-resistance-sampler}, and return
  \[
   Y_{\rm clip}=\frac{\rho S}{c_{st}}(a_{st}-W_{st}).
  \]
  The calculation of \Cref{lem:original-normalizer} gives
  $\mathbb EY_{\rm clip}=C_{\rm clip}$ and
  \[
   \mathbb EY_{\rm clip}^2
   =\rho S\sum_p\frac{w_p}{c_{st}}\mathbb E(a_{st}-W_{st})^2
   \le16\rho S\sum_pw_p
         \frac{\mathbb E(W_{st}-R_K(s,t))^2}{R_K(s,t)}
   =O(\lambda^2N^2).
  \]
  The expected sampling cost is $\widetilde O(1)$ by the argument of \Cref{lem:schur-two-stage-rejection}.
  
  After bounding the second moment of the three estimators, the remaining analysis is identical to the proof of \Cref{lem:original-normalizer}.
  \end{proof}
  
\appendix

\bibliographystyle{alpha}
\bibliography{ref}

\end{document}